\documentclass[11pt]{article}

\usepackage[letterpaper,margin=1.1in,marginparwidth=1.15in]{geometry}
\usepackage[T1]{fontenc}
\usepackage[utf8]{inputenc}
\usepackage{lmodern}
\usepackage{microtype}
\usepackage{setspace}
\usepackage{siunitx}
\makeatletter
\def\input@path{{sections/}{figures/}{./}}
\makeatother
\usepackage{graphicx}
\graphicspath{{figures/}}
\usepackage{booktabs}
\usepackage{threeparttable}
\usepackage{amsmath,amssymb,amsthm}
\usepackage{bbm}
\usepackage{caption}
\usepackage{xcolor}
\definecolor{NavyLink}{HTML}{10233F}
\definecolor{BrandBlue}{HTML}{3F7BD6}
\definecolor{BrandGreen}{HTML}{2FA877}
\definecolor{BrandGold}{HTML}{C9A11A}
\definecolor{BrandPurple}{HTML}{5F01D1}
\definecolor{SoftGray}{HTML}{7488A6}
\definecolor{LightBG}{HTML}{F2F5FA}

\usepackage{tikz}
\usetikzlibrary{positioning,arrows.meta,fit,backgrounds,calc,shapes.geometric,
                decorations.pathreplacing}

\usepackage[authoryear,round]{natbib}
\usepackage[colorlinks=true,allcolors=NavyLink]{hyperref}

\usepackage{authblk}

\newif\ifDRAFT
\DRAFTfalse                 %
\usepackage[colorinlistoftodos,textsize=scriptsize]{todonotes}
\ifDRAFT\else\presetkeys{todonotes}{disable}{}\fi

\newcommand{\gap}[1]{\todo[inline,color=red!12,bordercolor=red!60]%
                          {\textbf{OPEN.} #1}}

\theoremstyle{plain}
\newtheorem{proposition}{Proposition}

\theoremstyle{definition}

\theoremstyle{remark}
\newtheorem{remark}{Remark}

\newcommand{\pubs}{\mathcal{P}}
\newcommand{\blocks}{\mathcal{B}}
\newcommand{\E}{\mathbb{E}}

\newcommand{\credit}{\varphi}
\newcommand{\EQUAL}{\textsc{Equal}}
\newcommand{\PROV}{\textsc{Provenance}}
\newcommand{\COMPL}{\textsc{Completer}}

\newcommand{\NQueries}{1{,}994}          %
\newcommand{\NQueriesUsed}{1{,}974}      %
\newcommand{\NPublishers}{1{,}375}
\newcommand{\NPremium}{137}
\newcommand{\NPremiumPct}{10.0}
\newcommand{\NClusters}{64{,}419}
\newcommand{\NUrls}{7{,}553}             %
\newcommand{\NGroundings}{23{,}583}
\newcommand{\NBlockRows}{377{,}238}
\newcommand{\MeanCited}{10.66}           %

\newcommand{\MeanCovGlobal}{22.4}        %
\newcommand{\MeanCovLocal}{1.23}
\newcommand{\SingletonGlobal}{27.2}      %
\newcommand{\SingletonLocal}{86.2}
\newcommand{\SingletonNews}{10.5}
\newcommand{\SingletonComment}{49.3}

\newcommand{\ExactPrem}{33.3}
\newcommand{\ExactNonPrem}{58.4}
\newcommand{\AnyTimeAll}{87.2}
\newcommand{\NMidnight}{3{,}043}         %
\newcommand{\NNoTime}{966}

\newcommand{\EqGlobal}{41.87}
\newcommand{\ProvAsis}{49.55}
\newcommand{\ProvNoon}{49.11}
\newcommand{\ProvDayties}{46.79}
\newcommand{\ProvExact}{46.33}
\newcommand{\EqLocal}{53.75}

\newcommand{\ProvRange}{46.3\textrm{--}49.6}
\newcommand{\ProvLocalRange}{53.4\textrm{--}54.1}

\newcommand{\TokBoxPrompt}{16{,}000{,}243}

\newcommand{\TokTextPrompt}{35{,}067{,}592}

\newcommand{\TokPromptCut}{54.4}         %
\newcommand{\TokTotalCut}{52.3}          %
\newcommand{\TokRetained}{0.48}          %

\newcommand{\AlphaOne}{41.88}            %
\newcommand{\AlphaTwo}{45.11}
\newcommand{\AlphaFour}{47.30}
\newcommand{\AlphaEight}{48.64}
\newcommand{\ComplGlobal}{49.56}

\newcommand{\McNoise}{0.57}

\newcommand{\FourthGlobal}{47.57}

\newcommand{\RandArtGlobal}{45.26}

\newcommand{\TVArtPubGlobal}{5.57}
\newcommand{\TVArtPubLocal}{0.07}

\newcommand{\EqNews}{37.95}
\newcommand{\EqComment}{47.66}
\newcommand{\ProvNews}{48.51}
\newcommand{\ProvComment}{52.01}
\newcommand{\ComplNews}{42.24}
\newcommand{\ComplComment}{51.09}
\newcommand{\MeanCovNews}{35.1}
\newcommand{\MeanCovComment}{5.6}
\newcommand{\NBlocksNews}{11{,}991}
\newcommand{\NBlocksComment}{9{,}059}

\newcommand{\VetoGlobal}{66.3}
\newcommand{\VetoLocal}{97.1}

\newcommand{\NPubBase}{1{,}103}

\newcommand{\TopTenEqual}{23.2}
\newcommand{\TopTenProv}{34.6}
\newcommand{\TopTenCompl}{33.4}
\newcommand{\TopFiftyEqual}{50.1}
\newcommand{\TopFiftyProv}{68.5}
\newcommand{\TopFiftyCompl}{63.9}
\newcommand{\GiniEqual}{0.748}
\newcommand{\GiniProv}{0.885}
\newcommand{\GiniCompl}{0.844}
\newcommand{\ZeroEqual}{0.0}
\newcommand{\ZeroProv}{34.6}
\newcommand{\ZeroCompl}{12.3}

\newcommand{\GiniEqualLocal}{0.840}
\newcommand{\GiniProvLocal}{0.846}

\newcommand{\NYTEq}{0.62}
\newcommand{\NYTEqRank}{29}
\newcommand{\NYTPv}{4.70}
\newcommand{\NYTPvRank}{3}
\newcommand{\NYTRatio}{7.6}
\newcommand{\YahooEq}{2.95}
\newcommand{\YahooPv}{1.49}
\newcommand{\ReutersEq}{1.21}
\newcommand{\ReutersCm}{2.94}
\newcommand{\MilEq}{0.19}
\newcommand{\MilEqRank}{95}
\newcommand{\MilPv}{3.70}
\newcommand{\MilPvRank}{4}

\newcommand{\TVEqComplGlobal}{19.6}

\newcommand{\TVProvComplGlobal}{33.7}

\newcommand{\NPremiumUS}{54}             %
\newcommand{\NPremiumIntl}{83}           %
\newcommand{\NStandard}{1{,}238}
\newcommand{\NPremiumCountryUS}{55}
\newcommand{\NPremiumCountryFR}{15}
\newcommand{\NPremiumCountryIN}{13}
\newcommand{\NPremiumCountryUK}{9}
\newcommand{\NDupPublishers}{16}         %
\newcommand{\NTypeOther}{1{,}126}
\newcommand{\NTypeNewsroom}{149}         %

\newcommand{\EqTypeBased}{43.61}
\newcommand{\ProvTypeBased}{51.36}
\newcommand{\EqUSOnly}{20.09}
\newcommand{\ProvUSOnly}{28.70}
\newcommand{\EqIntlOnly}{21.79}
\newcommand{\ProvIntlOnly}{20.85}
\newcommand{\SpreadPooled}{7.68}
\newcommand{\SpreadTypeBased}{7.75}
\newcommand{\SpreadUSOnly}{8.61}
\newcommand{\SpreadIntlOnly}{-0.94}

\newcommand{\ErrMaxRate}{20}
\newcommand{\ErrSpreadLo}{$+7.62$}
\newcommand{\ErrSpreadHi}{$+8.12$}
\newcommand{\ErrMeanCovTwenty}{19}
\newcommand{\ErrKindFiveEq}{41.82}
\newcommand{\ErrKindFivePv}{49.56}
\newcommand{\ErrKindFiveSp}{$+7.74$}
\newcommand{\ErrKindTwentyEq}{41.66}
\newcommand{\ErrKindTwentyPv}{49.28}
\newcommand{\ErrKindTwentySp}{$+7.63$}
\newcommand{\ErrSplitFiveEq}{41.81}
\newcommand{\ErrSplitFivePv}{49.46}
\newcommand{\ErrSplitFiveSp}{$+7.65$}
\newcommand{\ErrSplitTwentyEq}{41.55}
\newcommand{\ErrSplitTwentyPv}{49.34}
\newcommand{\ErrSplitTwentySp}{$+7.79$}
\newcommand{\ErrMergeFiveEq}{41.89}
\newcommand{\ErrMergeFivePv}{49.66}
\newcommand{\ErrMergeFiveSp}{$+7.77$}
\newcommand{\ErrMergeTwentyEq}{41.97}
\newcommand{\ErrMergeTwentyPv}{50.08}
\newcommand{\ErrMergeTwentySp}{$+8.12$}

\newcommand{\MCMaxPerm}{800}
\newcommand{\MCReps}{three}
\newcommand{\MCSlopeTruth}{-0.503}
\newcommand{\MCSlopeRepLo}{-0.497}
\newcommand{\MCSlopeRepHi}{-0.517}
\newcommand{\MCTruthTwentyFive}{1.59}
\newcommand{\MCTruthTwoHundred}{0.56}
\newcommand{\MCTruthEightHundred}{0.28}
\newcommand{\MCRepRatio}{1.40}
\newcommand{\MCErrAlphaTwo}{0.50}
\newcommand{\MCErrAlphaEight}{0.74}
\newcommand{\MCErrCompl}{0.99}
\newcommand{\MCPremErr}{0.011}

\tikzset{
  every node/.append style={execute at begin node=\linespread{1}\selectfont},
  nd/.style      = {draw=NavyLink!55, fill=white, rounded corners=1.6pt,
                    inner sep=3pt, font=\scriptsize, align=center,
                    minimum height=4.6mm},
  bundle/.style  = {nd, fill=LightBG, draw=NavyLink!45, font=\small,
                    minimum width=26mm, minimum height=7mm},
  newsblk/.style = {nd, fill=BrandBlue!15, draw=BrandBlue!65},
  commblk/.style = {nd, fill=BrandGold!22, draw=BrandGold!85},
  pubnd/.style   = {nd, fill=white, draw=NavyLink!35, font=\scriptsize},
  panel/.style   = {draw=NavyLink!22, fill=LightBG, rounded corners=2.5pt,
                    inner sep=4pt, align=left, font=\scriptsize},
  lbl/.style     = {font=\scriptsize, text=SoftGray, align=center},
  tag/.style     = {font=\scriptsize\bfseries, text=NavyLink},
  ar/.style      = {-{Stealth[length=4.5pt,width=3.5pt]}, draw=NavyLink!65,
                    line width=0.55pt},
  arg/.style     = {ar, draw=BrandGreen!85},
  faint/.style   = {opacity=0.42},
  lego/.style    = {rectangle, rounded corners=1.2pt, line width=0.7pt,
                    align=center, inner sep=4pt, font=\scriptsize,
                    text=NavyLink},
  newslego/.style= {lego, fill=BrandBlue!15, draw=BrandBlue!70},
  commlego/.style= {lego, fill=BrandGold!22, draw=BrandGold!90},
  legoedge/.style= {-{Stealth[length=4pt,width=3pt]}, draw=BrandBlue!70,
                    line width=0.7pt, rounded corners=1.5pt},
  legoedgec/.style={legoedge, draw=BrandGold!90},
}

\newcommand{\legostuds}[4]{%
  \foreach \sdx in {#4}{%
    \node[rectangle, rounded corners=0.5pt, fill=#2, draw=#3, line width=0.5pt,
          minimum width=2.4mm, minimum height=1.3mm, inner sep=0pt,
          anchor=south] at ([xshift=\sdx mm]#1.north) {};}}
\newcommand{\newsstuds}[2]{\legostuds{#1}{BrandBlue!15}{BrandBlue!70}{#2}}
\newcommand{\commstuds}[2]{\legostuds{#1}{BrandGold!22}{BrandGold!90}{#2}}

\newcommand{\figblocktree}{%
\begin{tikzpicture}
  \node[bundle, minimum width=32mm] (art) at (0,2.75)
       {\footnotesize Article, publisher $p$};

  \node[newslego, text width=38mm] (n1) at (0,1.25)
       {Ceasefire announced 04:00 GMT};
  \newsstuds{n1}{-15,-5,5,15}

  \node[newslego, text width=28mm] (n2) at (-4.0,-1.20)
       {Two sites reported struck};
  \node[newslego, text width=28mm] (n3) at (0,-1.20)
       {Enrichment paused, says the IAEA};
  \node[commlego, text width=28mm] (c1) at (4.0,-1.20)
       {A fragile truce, settling nothing};
  \newsstuds{n2}{-11,-4,4,11}
  \newsstuds{n3}{-11,-4,4,11}
  \commstuds{c1}{-11,-4,4,11}

  \node[commlego, text width=28mm] (c2) at (0,-3.20)
       {Domestic politics drove the pause};
  \commstuds{c2}{-11,-4,4,11}

  \coordinate (j) at (0,0.42);
  \draw[ar] (art.south) -- (n1.north);
  \draw[draw=NavyLink!45, line width=0.7pt] (n1.south) -- (j);
  \draw[legoedge]  (j) -| (n2.north);
  \draw[legoedge]  (j) -- (n3.north);
  \draw[legoedgec] (j) -| (c1.north);
  \draw[legoedgec] (n3.south) -- (c2.north);
  \node[lbl, text width=32mm, anchor=north] at (-2.0,0.32)
       {one action may precipitate another};
  \node[lbl, text width=32mm, anchor=north] at (2.05,0.32)
       {commentary hangs from a news action};

  \node[newslego, text width=22mm] (kn) at (-4.3,-4.60) {\textbf{news}};
  \newsstuds{kn}{-7,0,7}
  \node[lbl, anchor=west, align=left, text width=50mm] at (-2.8,-4.60)
       {an action that changed the world};
  \node[commlego, text width=22mm] (kc) at (-4.3,-5.50) {\textbf{commentary}};
  \commstuds{kc}{-7,0,7}
  \node[lbl, anchor=west, align=left, text width=50mm] at (-2.8,-5.50)
       {an opinion, reaction, context or analysis on the news action it hangs from};
\end{tikzpicture}}

\newcommand{\figindexbuild}{%
\begin{tikzpicture}
  \foreach \i/\y/\p in {1/0/{Publisher A, 06:12}, 2/-2.0/{Publisher B, 07:40},
                        3/-4.0/{Publisher C, 11:05}} {
    \node[pubnd, anchor=west, minimum width=30mm] (p\i) at (-6.9,\y) {\p};
  }
  \node[newslego, text width=22mm, anchor=west] (a1) at (-3.4,0.42)  {ceasefire at 04:00};
  \node[commlego, text width=22mm, anchor=west] (a2) at (-3.4,-0.62) {a fragile truce};
  \node[newslego, text width=22mm, anchor=west] (b1) at (-3.4,-1.58) {ceasefire at 04:00};
  \node[newslego, text width=22mm, anchor=west] (b2) at (-3.4,-2.62) {two sites struck};
  \node[newslego, text width=22mm, anchor=west] (c1) at (-3.4,-3.58) {ceasefire at 04:00};
  \node[commlego, text width=22mm, anchor=west] (c2) at (-3.4,-4.62) {a win for hardliners};
  \foreach \n in {a1,b1,b2,c1} { \newsstuds{\n}{-8,-3,3,8} }
  \foreach \n in {a2,c2}       { \commstuds{\n}{-8,-3,3,8} }

  \foreach \s/\t in {p1/a1, p2/b1, p3/c1} { \draw[ar] (\s.east) -- (\t.west); }
  \draw[legoedgec] (a1.south) -- (a2.north);
  \draw[legoedge]  (b1.south) -- (b2.north);
  \draw[legoedgec] (c1.south) -- (c2.north);

  \node[panel, anchor=west, text width=42mm] (idx) at (0.9,-1.30)
    {\scriptsize\textbf{Index entry $b$}\\[1pt]
     \emph{ceasefire announced at 04:00}\\[2pt]
     publishers $P_b=\{A,B,C\}$\\
     first published 06:12 (A)\\
     kind: news};
  \foreach \s in {a1,b1,c1} { \draw[arg] (\s.east) -- (idx.west); }

  \node[panel, anchor=west, text width=42mm, fill=BrandGold!10,
        draw=BrandGold!40] (idx2) at (0.9,-4.30)
    {\scriptsize\textbf{Index entry $b'$}\\[1pt]
     \emph{a win for hardliners}\\[2pt]
     publishers $P_{b'}=\{C\}$\\
     kind: commentary};
  \draw[ar, draw=BrandGold!85] (c2.east) -- (idx2.west);

  \node[lbl, anchor=north, text width=112mm] at (-1.4,-5.55)
    {Built once, before any query arrives. Each publisher contributes a block
     tree; clustering runs across trees, so a claim carried by several outlets
     becomes one row listing all of them.};
\end{tikzpicture}}

\newcommand{\figanswering}{%
\begin{tikzpicture}
  \node[nd, fill=white, draw=NavyLink!45, text width=27mm] (q) at (-5.5,0)
       {\scriptsize ``Did the ceasefire hold, and who gained?''};
  \node[panel, text width=30mm] (ret) at (-2.0,0)
       {\scriptsize\textbf{Search engine}\\[1pt] retrieved grounding URLs which are mapped to lego trees from the index.};
  \node[panel, text width=30mm] (gen) at (1.9,0)
       {\scriptsize\textbf{Generate}\\[1pt] compose the answer, citing the
        block ids actually used};
  \node[nd, fill=BrandGreen!10, draw=BrandGreen!60, text width=30mm] (ans)
       at (5.6,0)
       {\scriptsize The truce held [b\,17, b\,42]; analysts read it as a
        domestic win [b\,88]};
  \draw[ar] (q) -- (ret);  \draw[ar] (ret) -- (gen);  \draw[ar] (gen) -- (ans);

  \node[lbl, anchor=north, text width=44mm] at (-0.05,-1.25)
       {the model does not see a publisher name, only block text and block ids};
  \node[lbl, anchor=north, text width=34mm] at (5.6,-1.25)
       {the citation list \emph{is} the usage record};
\end{tikzpicture}}

\newcommand{\figattribution}{%
\begin{tikzpicture}
  \node[nd, fill=BrandGreen!10, draw=BrandGreen!60, text width=26mm] (cited)
       at (-5.4,0) {\scriptsize cited blocks\\ $\{b_{17}, b_{42}, b_{88}\}$};

  \node[panel, text width=40mm] (tab) at (-1.7,0)
    {\scriptsize
     \begin{tabular}{@{}ll@{}}
       $b_{17}$ & $\{A,B,C\}$, first $=A$\\
       $b_{42}$ & $\{B,D\}$, first $=D$\\
       $b_{88}$ & $\{C\}$, first $=C$\\
     \end{tabular}\\[2pt]
     \emph{read straight off the index}};

  \node[panel, text width=33mm, fill=BrandPurple!7, draw=BrandPurple!35]
       (rule) at (2.5,0)
    {\scriptsize\textbf{Apply the rule}\\[1pt]
     pick an order over publishers and a utility over coverage; the credit
     vector follows};

  \node[nd, text width=24mm, fill=white] (pay) at (6.1,0)
    {\scriptsize \textbf{payout, \EQUAL}\\[2pt]
     \begin{tabular}{@{}lr@{}}
       $A$ & 0.33\\ $B$ & 0.83\\ $C$ & 1.33\\ $D$ & 0.50\\
     \end{tabular}};

  \draw[ar] (cited) -- (tab);
  \draw[ar] (tab) -- (rule);
  \draw[ar] (rule) -- (pay);

  \node[lbl, anchor=north, text width=54mm] at (0,-1.5)
    {No language model runs here. Settlement costs arithmetic, and the same
     inputs always produce the same payment --- so the ledger can be
     audited by anyone holding the index.};
\end{tikzpicture}}

\newcommand{\figdials}{%
\begin{tikzpicture}
  \node[tag] at (0,2.95) {order over publishers};
  \node[lbl] at (-1.95,2.5) {uniformly random};
  \node[lbl] at ( 1.95,2.5) {chronological};

  \node[tag, rotate=90] at (-6.7,0) {utility over coverage};
  \node[lbl, anchor=east, text width=20mm] at (-4.3, 1.3)
       {linear\\ $u(x)=x$};
  \node[lbl, anchor=east, text width=20mm] at (-4.3,-1.3)
       {all-or-nothing\\ $u(x)=x^{\alpha}$,\\ $\alpha\to\infty$};

  \node[panel, text width=30mm, minimum height=20mm, align=left,
        fill=BrandBlue!10, draw=BrandBlue!45] at (-1.95, 1.3)
       {\scriptsize\textbf{\EQUAL}\\ split each block equally among its
        publishers\\[2pt] \emph{the Shapley value}};
  \node[panel, text width=30mm, minimum height=20mm, align=left,
        fill=BrandGold!14, draw=BrandGold!55] at ( 1.95, 1.3)
       {\scriptsize\textbf{\PROV}\\ the first publisher takes the whole
        block\\[2pt] \emph{a scoop premium}};
  \node[panel, text width=30mm, minimum height=20mm, align=left,
        fill=BrandGreen!12, draw=BrandGreen!45] at (-1.95,-1.3)
       {\scriptsize\textbf{\COMPL}\\ credit lands on whoever completes the
        answer\\[2pt] \emph{rewards the rare block}};
  \node[panel, text width=30mm, minimum height=20mm, align=left,
        fill=BrandPurple!7, draw=BrandPurple!35] at ( 1.95,-1.3)
       {\scriptsize\textbf{fourth corner}\\ first publisher, but only where
        coverage is thin\\[2pt] \emph{implied by the dials}};
\end{tikzpicture}}

\title{\vspace{-3.5em}\bfseries Paying for Information, Not for Clicks:\\[0.2em]
       A Block-Level Publisher Content Marketplace for AI Answers}

\author[1]{Markus Mobius}
\author[6]{Luca Braghieri}
\author[6]{Sarah Eichmeyer}
\author[1]{Daniel G.~Goldstein}
\author[1]{Nicole Immorlica}
\author[2]{Oliver Kiss}
\author[7]{Ro'ee Levy}
\author[1]{Brendan Lucier}
\author[3]{Ishan Mahajan}
\author[4]{Jingyi Qiu}
\author[4]{Tanya Rosenblat}
\author[1]{David M. Rothschild}
\author[5]{Andras Sari}
\author[1]{Aleksandrs Slivkins}
\author[3]{Colin Watson}

\affil[1]{Microsoft Research}
\affil[2]{Spreadmonitor}
\affil[3]{Microsoft}
\affil[4]{University of Michigan}
\affil[5]{Central European University}
\affil[6]{University of Bocconi}
\affil[7]{Tel Aviv University}

\date{\vspace{-2.2em}\today\\[0.3em]\normalsize\textbf{Working paper}}

\begin{document}
\maketitle
\vspace{-3.1em}

\renewcommand{\thefootnote}{\fnsymbol{footnote}}
\footnotetext[1]{Mobius is first author.  All other authors are listed alphabetically. The lego decomposition
approach was developed by Braghieri, Eichmeyer, Kiss, Levy, Mobius, Rosenblat and Sari. Qiu collaborated with this group on adapting the decomposition into the lego implementation. Goldstein, Immorlica, Lucier, Mobius, Rosenblat, Rothschild and Slivkins developed
the theory framework for measuring attribution.  Mahajan and Watson helped frame the
problem.  All errors are our own.}
\renewcommand{\thefootnote}{\arabic{footnote}}
\setcounter{footnote}{0}

\begin{abstract}
\singlespacing\small
\noindent

We propose a publisher content marketplace that allocates payments for information used in AI-generated answers. The unit of attribution is a ``lego block'': a distinct fact or piece of commentary that may appear in articles from multiple publishers. A simple news-reading model connects three payment rules: equal sharing among all publishers reporting on a block, rewarding the publisher who reported a block first, and rewarding indispensable contributions. Each can be represented within a Shapley-value framework.

We evaluate the framework using answers to \NQueriesUsed{} search-engine queries about a heavily covered news episode. When attribution considers only the articles retrieved for an individual query, most cited blocks appear unique, and the rules produce nearly identical allocations. A \textit{global index} covering the full news corpus reveals overlap across publishers, making the choice of payment rule consequential. Premium publishers, defined as established outlets with substantial newsgathering operations, receive \num[round-mode=places,round-precision=0]{\EqGlobal}\% of allocated value under equal sharing, compared with \num[round-mode=places,round-precision=0]{\ProvExact}--%
\num[round-mode=places,round-precision=0]{\ProvAsis}\% under a first-publication rule, depending on the treatment of publication timestamps. These findings highlight the importance of a shared global index that records which publishers provide each block and when they published it. Attribution restricted to the sources retrieved for an individual answer lacks the information needed to distinguish these payment policies.


\end{abstract}

\clearpage
\section{Introduction}
\label{sec:intro}

Historically, publishers of online content were paid in clicks. An intermediary (e.g., search engine or a social feed) showed a headline, the reader clicked through to the publisher's page, and the publisher 
monetized the visit. The intermediary profited, too, receiving a payment for the referral to the publisher.

Generative search breaks this arrangement. When an assistant answers a question directly, 
it summarizes the publisher's reporting while providing the publisher almost no traffic. \citet{pew2025} 
find that users click through on 8\% of searches carrying AI summaries (compared to 15\% without), 
and click a link \emph{inside} the summary on just 1\%. 
Yet the underlying content is more valuable to the intermediary than ever, as a grounded answer cannot be generated without it.

This shift has created a market-design problem that platforms are actively attempting to solve: finding alternative channels through which to value content and pay publishers. For example,  
Cloudflare's Pay-Per-Crawl marketplace~\citep{allen2025} and Amazon's AWS AI traffic monetization system~\citep{kayabali2026}, unofficially dubbed the ``AWS AI Content Marketplace,'' allow publishers to charge AI crawlers for content access. While these systems allow publishers to set prices, they do not directly assist with price discovery or value determination.
In February 2026, Microsoft announced a Publisher Content Marketplace to pay 
publishers on ``delivered value'' through usage-based reporting, launching with partners 
such as the Associated Press, Cond\'e Nast, Vox Media, and the USA TODAY Company 
\citep{frank2026}. While this announcement establishes the \emph{ambition} of paying 
for the content that's actually used, it leaves the underlying \emph{mechanism} open.

Why is it challenging to  monetize content directly?
The difficulty begins with the unit of account. Media consumption has unbundled over two decades. News used to be consumed in newspapers consisting of collections of articles.  With the rise of the internet, that was unbundled into individual articles, and then rebundled through news aggregators. 
Music underwent 
a similar shift from album to track to multi-source playlist, re-pricing the underlying media at every step. 

With AI, there is an additional unbundling of content, from full coherent pieces like news articles, to their components -- facts or comments.  The market for publisher content must adjust its pricing model accordingly: if a marketplace's unit of 
account remains the article or webpage, it measures something that is no longer consumed.
This suggests monetizing at a more granular level, but facts and comments are difficult to price. They are \emph{non-rival} (one reader's consumption does
 not deplete the fact or comment) \citep{jones2020} and \emph{non-excludable} (once published, anyone can legally re-report it) \citep{arrow1962}. 
 In practice, non-excludability yields rampant \emph{duplication}. In our data, a cited news 
 fact appears in the reporting of a mean of \MeanCovGlobal{} distinct publishers: one filed first, 
 while others carried the same fact through independent reporting or aggregation.

This paper makes four core contributions:

\paragraph{Contribution 1: A unit of account for news and answer generation.}
We propose a fine-grained, interpretable unit of account for news content. News reporting can be broken down into atomic \emph{``lego'' blocks}: \emph{news} blocks record state-changing actions in the world, while \emph{commentary} 
blocks record stances, opinions, or context regarding those actions. 

This unit offers two primary advantages over existing approaches. It is \emph{explainable}: 
blocks are short declarative sentences that publishers can audit directly without relying on obscure 
model saliency scores or counterfactual re-generations\footnote{Counterfactual re-generation techniques estimate source influence at inference time 
by repeatedly removing retrieved documents from the context window, re-running the language model, 
and measuring changes in the output text (see, e.g., \citealp{cohenwang2024, qi2024}).}. It is also \emph{comparable}: a single comment or fact, reported by many distinct publishers, is a single block, reducing redundancy.

\paragraph{Contribution 2: A system for indexing articles and answering queries with attribution.}

We describe a system that constructs a coherent index of lego blocks from a corpus of news articles. 
Each crawled article is decomposed 
by a language model into constituent blocks and arranged into a tree that represents their relationships.
Identical blocks across publishers 
are clustered and deduplicated into a shared, tagged index that records every atomic claim, its publishers, and their respective timestamps. This generates a shared global index of lego blocks that represent the content of multiple articles.

We demonstrate 
that building this representation at index 
time simplifies the query response process at inference time, allowing queries to be settled via straightforward table lookups. Furthermore, 
generating answers from deduplicated blocks significantly reduces 
token overhead by inserting shared facts or comments into the model's context window once rather than repeatedly. Finally, since the index retains publisher information for each block, provenance is retained and credit attribution can be calculated whenever an answer is generated.  In contrast, an index constructed at inference time retains only knowledge of a single source of a given block but not the collection or reporting times.  Thus computing attribution at inference time, e.g., by perturbing or inspecting retrieved 
context post-query \citep{cohenwang2024, qi2024, nematov2025},  
cannot distinguish an originator from a copier because that global context was discarded prior to query execution, limiting its ability to implement sophisticated attribution rules.

\paragraph{Contribution 3: A unifying economic theory of reader value and credit attribution.}
A primary motivation for credit attribution is determining how to appropriately divide payments among publishers for their content.
Once news reports map to lego blocks, the payment problem simplifies: for every block 
used in an answer, the index identifies who reported it and when. This information can be used as input to a \emph{payment rule} that compensates publishers.

We describe a flexible space of payment rules that can be used. We evaluate three intuitive payment rules drawn from existing literature:
\begin{enumerate}
\item Equal division among all publishers that carried the block (\EQUAL{}), mirroring music-streaming pro-rata models \citep{datta2018};

\item Exclusive payment to the originator (\PROV{}), addressing the under-rewarding of original scoops \citep{cage2020restud};

\item Payment based on indispensability (\COMPL{}), following the default posture of data valuation via the Shapley value \citep{ghorbani2019, jia2019}.
\end{enumerate}

We show that all three rules can be interpreted as Shapley values for different formulations of a content-creation game.
Imagine a reader who reviews grounded articles in sequence and registers new blocks as she reads while ignoring blocks she already encountered. 
Each publisher is compensated for the marginal utility added when the reader reads their article.
Setting the reading order to \emph{random} and reader utility to \emph{linear} in the facts discovered yields \EQUAL{}; 
setting the order to \emph{chronological} yields \PROV{}; and holding the order random while 
pushing utility toward an \emph{all-or-nothing} model yields \COMPL{}. This framework provides a 
structured foundation for evaluating payment mechanisms, mapping them as points along two 
parameterized economic dimensions: the \emph{order dial} (whether speed deserves reward) and 
the \emph{utility dial} (how heavily completeness is valued).\footnote{The model is also generative: setting both dials to \emph{chronological} and \emph{all-or-nothing} 
reveals an unstudied fourth rule that compensates the \emph{last} publisher needed to complete an answer. 
We view its emergence as a by-product of the two dials rather than a proposal.}

\paragraph{Contribution 4: Empirical findings and index-time necessity.}
We implement and test our system against a broad corpus of news articles related to the Iran/US war in the summer of 2026. 
We find that basing attribution on a global index rather than a local one has a large empirical impact. Against a global index, only \SingletonGlobal\% of cited blocks have a single publisher, and cited blocks have an average of \MeanCovGlobal\% publishers. 
We find that different payment rules separate clearly, shifting premium publishers' revenue share from \EqGlobal\% under \EQUAL{} to \ProvRange\% under \PROV{}. 
Against a local index, however, \SingletonLocal\% of cited blocks appear unique within the retrieval set (with a mean coverage of \MeanCovLocal{}), causing every payment 
rule in the family to collapse to the same payout within one percentage point. This suggests that credit attribution is most meaningful when evaluated against a global pre-constructed index:
attribution rules matter only when duplication is visible, and duplication is only visible globally. We also find that the \EQUAL{} rule raises the relative value of commentary over original news reporting, due to commentary blocks being less duplicated than news blocks. This suggests that a credit attribution system based on provenance can meaningfully help to reward original reporting by premium publishers in a content marketplace. On the other hand, we find that the \PROV{} rule concentrates payments toward a small number of large originators, away from the tail of smaller publishers.

\paragraph{Roadmap} The paper proceeds as follows. Section~\ref{sec:literature} reviews related literature. 
Section~\ref{sec:lego} describes the index, its construction, and its use in answer generation. 
Section~\ref{sec:mechanism} describes attribution scoring options and develops the unified Shapley mechanism. 
Section~\ref{sec:empirical} describes our datasets and %
empirical results, 
and Section~\ref{sec:conclusion} concludes.

\section{Related literature}
\label{sec:literature}

This paper proposes a market-design approach to news markets in an AI-centric world.  As such, the research relates to prior work on the economics of news, referral systems, product unbundling, information goods, data valuation, generative attribution, and relevant legal frameworks.  We discuss each in turn.

\paragraph{Economics of News Reporting.}
A long tradition studies how news market structures shape reporting incentives \citep{mullainathan2005, gentzkow2010, gentzkow2015, prat2013}.
Closest to our work is \citet{cage2020restud}, who document that a significant share of online news is aggregated or 
copied rather than originally reported, leaving original reporting under-rewarded as copies appear 
almost instantaneously. Their finding provides the empirical motivation for this paper. Our framework 
measures this phenomenon at the level of individual facts rather than entire articles: we show that a 
cited news block is carried by a mean of \MeanCovGlobal{} publishers but originated by 
just one, and our order dial explicitly dictates whether that head start is monetized. 
Related work by \citet{cage2020jpube} and \citet{angelucci2024} explores the broader implications for media competition and public information diets.

\paragraph{Aggregators and the Unbundling of Media.}
How news content is monetized has been studied extensively through the lens of search engines and news aggregators, 
focusing on the traditional referral-traffic bargain \citep{athey2021, chiou2017, calzada2020, aguiar2021}. 
Generative search changes this dynamic fundamentally: rather than reducing referral traffic, it eliminates it \citep{pew2025, newman2025}. 
Consequently, the market-design challenge shifts from analyzing bargaining power to establishing a precise metering standard. 
Unbundling content from published bundles to atomic components is a well-established economic trajectory \citep{bakos1999}, 
seen most clearly in the music industry's transition from albums to tracks \citep{datta2018}. Music streaming 
faced a nearly identical challenge---distributing a shared subscription pool across rights holders---and settled on pro-rata and 
user-centric allocation rules that closely mirror some of the credit attribution rules we consider. 

\paragraph{Economics of Information Goods.}
Underpinning this structure is the economics of non-rival information goods \citep{jones2020}, where market 
transactions are famously constrained by Arrow's Information Paradox: a buyer cannot evaluate information without 
consuming it, yet once consumed, cannot ``un-see'' it \citep{arrow1962, bergemann2019, bergemann2022}. A block-level index 
addresses this paradox by making content inspectable in granular units that can be priced without exposing the full article text. 
More broadly, our approach aligns with the normative framework of treating data as labor in digital economies \citep{arrieta2018}.

\paragraph{Technical Data Valuation.}
Technically, our framework connects to the literature on data valuation in machine learning. 
The Shapley value has emerged as the standard theoretical instrument \citep{ghorbani2019, jia2019, wang2025}, 
often estimated via influence functions or datamodels due to computational constraints \citep{koh2017, ilyas2022, park2023},
 with applications extending to royalty schemes for generative platforms \citep{wang2024copyright}. Crucially, however, 
 that literature values data for \emph{training}; we value content for \emph{serving} (retrieval augmented generation). 

\paragraph{Attribution vs. Settlement in Generative AI.}
Generative attribution literature serves as the primary contrast class for Section~\ref{sec:globalindex}. Current technical 
approaches evaluate citations at inference time \citep{rashkin2021, bohnet2022, gao2023, liu2023} or attribute generated text 
via context ablation, perturbation, or internal model introspection \citep{cohenwang2024, qi2024, goldshmidt2024, nematov2025}. 
Our claim is not that these methods fail at their intended task, but rather that their task is answer \emph{verification}, 
not financial \emph{settlement}. Because inference-time methods are constrained to a query's local retrieval set by construction, 
they are structurally incapable of executing content compensation policies---a limitation we formalize in Section~\ref{sec:globalindex} 
and quantify in Section~\ref{sec:results}. One exception is ProRata.ai~\citep{businesswire2024prorata}, which employs attribution methods on generated text to assign fractional credit and compensation to different data sources. Their approach is notably post-generation, applying to output text and external to the generation process, whereas our proposed marketplace attributes credit for data retrieved during generation itself. Complementary systems literature on passage granularity \citep{chen2023} and 
corpus deduplication \citep{lee2022} provides additional engineering support for our block-level representation.

\paragraph{Legal and Policy Context.}
Finally, legal frameworks surrounding fair use, content licensing, and copyright consent are evolving rapidly 
\citep{lemley2021, henderson2023, sag2023, longpre2024}. We take no stance on legal doctrines, but emphasize 
that regardless of the legal theory adopted, almost every proposed remedy requires a robust, fine-grained measurement 
layer that does not currently exist.

\section{The Lego Index}
\label{sec:lego}

This section describes the global index of lego blocks, which serves as the settlement layer of our proposed marketplace. The index 
represents a corpus of articles in a form that makes answer retrieval and economic attribution both feasible and 
computationally efficient. 
The index 
is a relational table of atomic claims, each tagged with the outlets that reported it and their respective publication timestamps. 
Storing these two fields directly, rather than recomputing them per query, reduces attribution to a  
table lookup.

\subsection{Decomposing an article into a block tree}

Our index decomposes the text of each article in a corpus into a tree of atomic \emph{blocks}. 
Blocks fall into two distinct functional categories:

\begin{description}
    
\item[News blocks] record an \emph{action}: a verifiable state change in the world brought about by individuals, organizations, 
or nature (e.g., a vote, a strike, an attack, or a hurricane). They document who acted, what occurred, when, and where. 
Because these facts form the skeleton of news articles, they are frequently duplicated across outlets.

\item[Commentary blocks] record an opinion or reaction (typically through a quote for non-opinion pieces), piece of context, or analytical stance \emph{about} a news action. 
They tend to be less duplicated across publishers and are often optional: the reporting of the underlying event can stand without them.
\end{description}

Separating news from commentary  isolates the two economically distinct goods a publisher sells.
Facts tend to be reported by multiple publishers, so provenance-based attribution methods significantly impact them.  Commentary, on the other hand, is often unique to a single publisher, causing it to play a different role in attribution.  Section~\ref{sec:results} investigates the differential impact of attribution rules on facts and commentary.

Extracted blocks are organized into a structured \emph{tree} rather than an unorganized bag of claims. 
Each commentary block is attached to the specific news block it references. Similarly, a news block may logically follow a preceding news block 
that precipitated it (e.g., a ceasefire hanging below the strikes that preceded it). The decomposition of an article into a tree of blocks takes several language-model passes and runs as soon as the search index crawler discovers a new article.  Figure~\ref{fig:blocktree} illustrates a worked example.

We maintain this tree structure for two reasons.  First, it is essential for clustering and deduplication, as we describe in Section~\ref{sec:clustering}.
Second, the tree representation serves a vital role during answer generation, as we discuss in Section~\ref{sec:answering}.  Furthermore, although we do not address it here, attribution methods for certain types of content, such as academic papers, might wish to provide attribution to blocks that are predecessors of cited ones.

\begin{figure}[tbp]
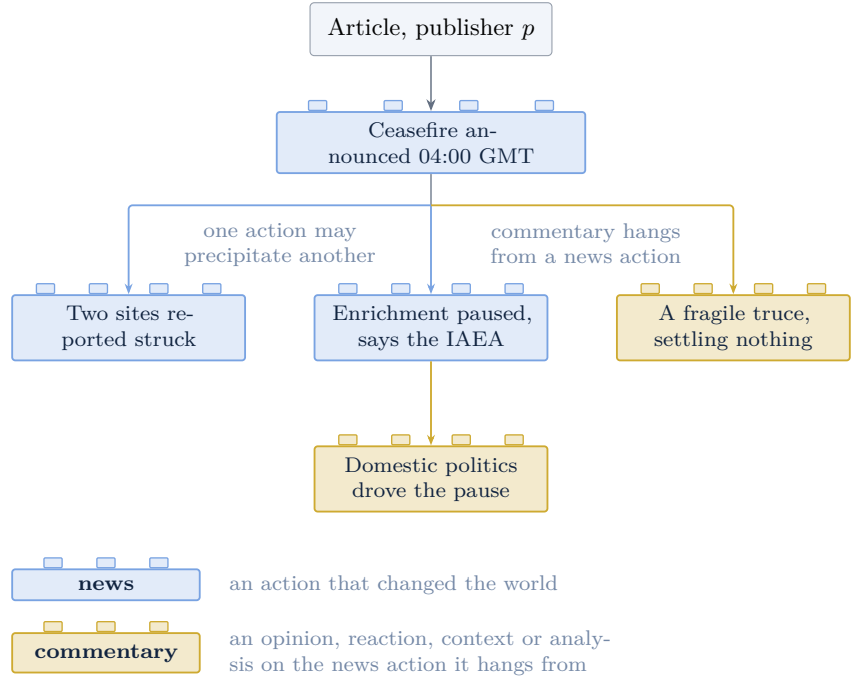
\centering\figblocktree\caption{\textbf{A block tree.} One article 
decomposed into news and commentary blocks. 
}
\label{fig:blocktree}
\end{figure}

\subsection{Clustering across publishers}
\label{sec:clustering}

Once a block trees have been constructed for each individual article, they are combined into a unified index.
Blocks expressing equivalent semantic content are clustered \emph{across} and \textit{within} publishers, mapping a block reported by twenty outlets 
to a single index row with twenty associated publishers rather than twenty separate entries. This clustering and deduplication process relies heavily on the tree structure constructed for each article. In isolation, a commentary sentence like \emph{``This decision was a mistake''} cannot be meaningfully matched across publishers: two such statements 
in different outlets might refer to entirely different news actions. The parent node establishes the necessary context, so deduplication clustering necessarily evaluates a commentary block together with its parent news block.

These two steps---transforming an article into a lego tree of blocks and then deduplicating those blocks---transform a list of URLs into the 
\textit{global index}. Each cluster row records every contributing publisher along with its publication timestamp. 
Figure~\ref{fig:indexbuild} traces this ingestion process: shared claims stack additional publisher tags on existing rows, while novel claims append new rows.

\begin{figure}[tbp]
\centering
\figindexbuild
\caption{\textbf{Building the index.} Each publisher contributes a block tree; blocks from separate trees that express the same claim collapse into a single index row. The row records who carries the claim and who filed first.}
\label{fig:indexbuild}
\end{figure}

\subsection{Properties, scope, and limitations}

The global index is a representation of a news corpus, optimized for attribution. 
As such, the index is bounded by the reach of the corpus upon which it is based.
Furthermore, the index measures \emph{coverage}, not \emph{veracity}. It records that an outlet published a claim, not whether the claim is factual, and 
no payment rule in this paper conditions payout on truthfulness. In practice, this distinction is less restrictive than it might appear. 
Across the broader media ecosystem, fabricated news represents a tiny fraction of total consumption \citep{allen2020, grinberg2019, guess2020}, meaning most 
news blocks record real-world events. Commentary blocks, by contrast, frequently quote external statements where accurate reporting and factual accuracy diverge. 
The index faithfully records the reporting of the claim; distributing payments based on verified truth 
would require an auxiliary verification signal outside the index itself. In any case, accurately reporting a false statement made by someone else is
not fabrication, and the commentary block is a correct record of the reporting either way. 

The global index can easily be made inspectable: publishers could directly inspect their indexed lego blocks, verify timestamp records, and audit the set of co-reporting publishers. 
We expect this transparency to be a prerequisite for any practical, functional content marketplace, as publishers will retain the ability to audit metrics upon which their payments are based.

\subsection{Answering queries and attributing credit using the index}
\label{sec:answering}

When a query arrives, a search engine retrieves a grounding set of relevant articles. 
The choice of grounding set is optimized by the retrieval system and is not under the direct control of our system; it may or may not include all applicable articles, but is designed to include sufficient coverage to address the given query. %
In a conventional retrieval-augmented generation (RAG) pipeline, 
the raw text of those articles is placed directly into the model's context window, and the model 
generates an answer accompanied by URL citations. In our pipeline, the retrieved articles 
are instead represented by their pre-computed lego block trees. The model is prompted to 
synthesize its response using these deduplicated blocks and to explicitly cite the unique block 
identifiers ($b$) it relied upon. Figure~\ref{fig:answering} illustrates this query pipeline.

\begin{figure}[tbp]
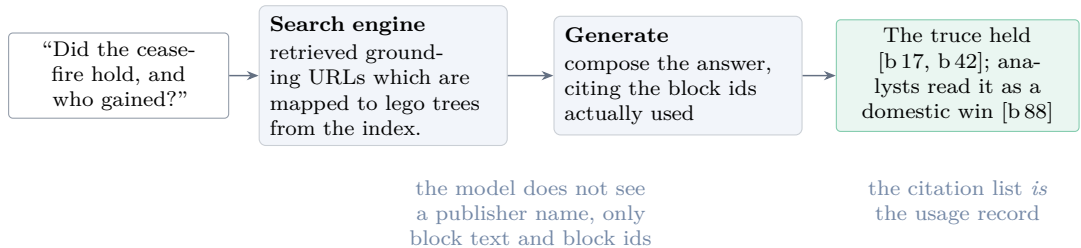

\centering
\figanswering\caption{\textbf{Answering Queries.} A query arrives, deduplicated block trees rather than raw article text are placed into context, and the model cites specific block identifiers in its response.}\label{fig:answering}
\end{figure}

Deduplicating facts before insertion into the 
context yields a computational benefit by reducing context length. Constructing a global index therefore does not inflate query-time generation costs; rather, by removing redundant 
facts prior to context assembly, it makes answer generation cheaper. We evaluate these savings empirically in Section~\ref{sec:results}.

Finally, credit for the answer must be divided among publishers. 
Fortunately, this credit attribution does not impose a significant extra calculation cost on the answer generation process. Because the model explicitly outputs cited block identifiers, and each block is associated with a list of articles to which it can traced, attribution requires no additional post-hoc LLM calls or context perturbations. Furthermore, publishers reporting a block can receive attribution even if they are not returned by grounding as the global index contains a record of their report alongside the block itself.
Settlement reduces to a lookup against the indexed articles. What remains is to chose how attribution is calculated, which is detailed in Section~\ref{sec:mechanism}.

\subsection{Scope and Timing}

We emphasize two aspects of credit attribution in our system.  First, it is evaluated against the global index of all publishers, as opposed to a local index consisting only of publishers represented in the grounding set returned for a given query.  Second, attribution of credit for each block is computed at index time, while the global index is being built; not at inference time, at the point where a query is to be answered.  

Attribution performed at index time and at a global scope requires an up-front cost to build and maintain the index, but saves on per-query computation costs and provides a more holistic view of the content landscape from which to evaluate provenance.  This distinguishes our method from various inference-time attribution approaches, such as asking the model to cite its sources
\citep{gao2023, bohnet2022}; perturbing the context and re-running, as in
ContextCite \citep{cohenwang2024} or Shapley over retrieved documents
\citep{nematov2025}; and reading the model's internals
\citep{qi2024}. While these methods are effective at explaining and verifying generation, we feel that global index-time attribution is desirable for credit attribution.  We discuss this comparison further in Appendix~\ref{sec:globalindex}.

\section{Credit Attribution Methods}
\label{sec:mechanism}

How should credit for an answer be divided among publishers? There are three natural credit attribution rules:
divide credit equally among every publisher associated with a block cited in the answer;
assign credit to whomever reported the block earliest, or give credit for marginal contribution over what others could provide.
This section shows that all three emerge from a
stylized model of news reading as we vary the model parameters.
Section~\ref{sec:shapley} then shows that these attribution rules can be interpreted as Shapley values of variations of a content generation game.

\subsection{The three candidate rules}
\label{sec:rules}

Fix a query. Let $\blocks$ be the set of blocks cited in the answer,
$n=|\blocks|$, and for each block $b\in\blocks$ let $P_b\subseteq\pubs$ be
its \emph{covering set}: the publishers whose articles contain that block.
Let $t_{p,b}$ be the time at which publisher $p$ published its article
carrying $b$. One unit of value is to be divided among publishers for
each query. An attribution rule $\credit$ assigns to each publisher $p$ a share of this value, $\credit_p \in [0,1]$.  We focus on the following three candidate rules.

\begin{description}
\item[\EQUAL{}.] Each block is worth $1/n$ and is split evenly among its
  publishers:
  $$\credit^{\EQUAL{}}_p=\frac{1}{n}\sum_{b\in\blocks}\frac{\mathbbm{1}\{p\in P_b\}}{|P_b|}.$$
\item[\PROV{}.] Each block's value goes entirely to whoever published it
  first, with ties split evenly:
  $$\credit^{\PROV{}}_p=\frac{1}{n}\sum_{b\in\blocks}\frac{\mathbbm{1}\{p\in\arg\min_{q\in P_b}t_{q,b}\}}{|\arg\min_{q\in P_b}t_{q,b}|}.$$
\item[\COMPL{}.] Each publisher is, loosely speaking, credited for the unique blocks it contributes.
  \citep{ghorbani2019, jia2019}. Unlike the
  other two, this rule is typically defined as a Shapley value rather than expressed in closed-form.
  Consider the cooperative game in which a coalition $S\subseteq\pubs$ has value $v_{all}(S)=\mathbbm{1}\{S\cap P_b\neq\emptyset\ \text{for all }
  b\in\blocks\}$ if it can supply \emph{every} block of the answer and $0$
  otherwise.  Let $\sigma$ be a permutation of publishers, $S_{\sigma(p)}=\{p':\sigma^{-1}(p')<\sigma^{-1}(p)\}$ be the publishers that appear before $p$ in $\sigma$,\footnote{Note $S_{\sigma(p)}$ can equivalently be defined as a random subset of publishers excluding $p$; the definition in terms of permutations is helpful for our characterization of attribution rules in a common framework and also useful for computing the Shapley value.} and let
  \[
\credit_p^{\COMPL{}}=E_\sigma[v_{all}(S_{\sigma(p)}\cup\{p\}])-v_{all}(S_{\sigma(p)})]
\]
   be the Shapley value of publisher $p$. An
  answer that is missing anything is worth nothing, so a publisher earns
  only in those permutations in which it is the one that completes
  the set. The rule has no closed form and we evaluate it by sampling
  permutations.
\end{description}

\begin{remark}
\label{sec:reservation}

Because $v_{\mathrm{all}}$ is equal to $1$ only when
\emph{every} block is covered, 
publishers that supply  duplicated blocks 
only receive credit for such a block in permutations in which they are the first among the duplicate publishers.
The residual can be small, as we report in Section~\ref{sec:concentration} on our empirical data. 
The rule is therefore close to being a
rule that pays for exclusivity alone.
Exclusivity is often a knife-edge property of the index rather than of the world. Adding just one outlet carrying an otherwise exclusive block removes an existing publisher's exclusivity; failing to merge two wordings of the same fact into a single block  can create exclusivity. Under \EQUAL{} an error of that kind moves a credit allocation from
$1/|P_b|$ to $1/(|P_b|\pm1)$ and washes out across blocks; under \COMPL{}
it is the difference between a large share and nothing. This also makes the rule susceptible to manipulation, since the way to become indispensable is to
publish what nobody else will match, which rewards being hard to cluster
rather than being first or being right.  For these reasons, we view the completer rule as inherently less robust than the other options.
\end{remark}

\subsection{A Unifying Framework: The Reading Model}
\label{sec:readingmodel}

Consider a stylized reader---think of her as the buyer's proxy---who is
handed the grounded sources for a query and works through them in some
order $\pi$. She begins knowing nothing. Each time she finishes a source
she registers whichever of its blocks are new to her; blocks she has
already seen elsewhere add nothing. Her satisfaction with what she has
learned so far is $u(x)$, where $x\in[0,1]$ is the fraction of the
answer's blocks she has covered, with $u(0)=0$ and $u(1)=1$. We take
\begin{equation}
u(x)\;=\;x^{\alpha},\qquad \alpha\ge 1,
\label{eq:utility}
\end{equation}
a one-parameter family that spans the range of views a buyer might hold
about an incomplete answer. At $\alpha=1$ utility is linear and every
additional fact is worth the same. As $\alpha$ rises the early facts count
for less and the last ones for more, and in the limit $\alpha\to\infty$ we
have $u(x)\to0$ for every $x<1$ while $u(1)=1$: the buyer has no value for the answer before consuming all its blocks. The exponent therefore can be interpreted as controlling \emph{how much the buyer values complete answers} and it is a parameter of our model.

Each of the aforementioned attribution rules can be interpreted as assigning a publisher credit for the increment in the reader's satisfaction that it
caused. Writing $c_p(\pi)$ for the number of blocks already covered when publisher
$p$ is reached under order $\pi$,\footnote{
One might read the sources in a random order over \emph{articles} rather
than over \emph{publishers}, which is not the same thing when a publisher
contributes several articles. In our data the two differ by 0.14\% of
total mass, so we use the publisher-level version throughout and report
the article-level version in Appendix~\ref{app:robustness}.
} and $a_p(\pi)$ for the number of blocks
$p$ is the first to supply,
\begin{equation}
\credit_p \;=\; \E_{\pi}\!\left[\,
  u\!\left(\frac{c_p(\pi)+a_p(\pi)}{n}\right)
  -u\!\left(\frac{c_p(\pi)}{n}\right)\right].
\label{eq:unified}
\end{equation}
Equation~\eqref{eq:unified} describes the mechanism. It has exactly two
free parameters (or ``dials"): the distribution over reading orders $\pi$, and the
exponent $\alpha$. These two parameters determine the credit attribution rule. %
Figure~\ref{fig:dials}
illustrates the construction.

\begin{figure}[tbp]
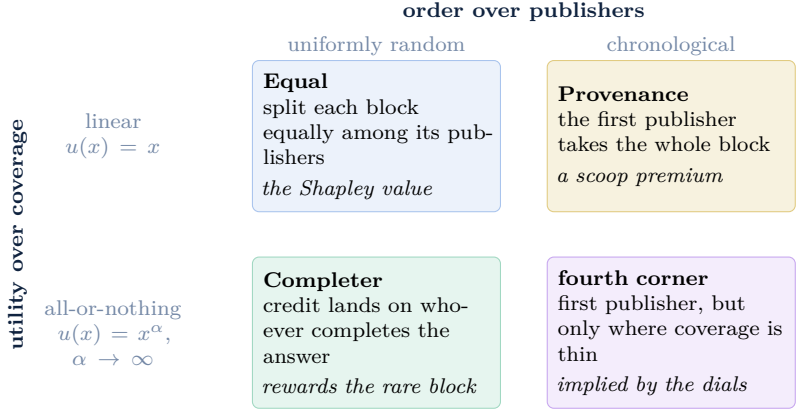
\centering
  \figdials
  \caption{\textbf{Credit Allocation Rules.} Methods for dividing credit among publishers who contributed a duplicated block, expressed as a function of reading order and utility curvature. Three of
  the corners are the rules described in Section~\ref{sec:rules}; the fourth is implied by the framework.}\label{fig:dials}
\end{figure}

\subsection{Varying the order parameter: distribution over $\pi$}
\label{sec:orderdial}

First fix the exponent $\alpha=1$ so that utility is linear in the fraction of blocks read. If the reading order $\pi$ is chosen uniformly randomly over all permutations, a block is credited to whichever of its
publishers the reader happens to reach first, and by symmetry that is each
of them with probability $1/|P_b|$. In expectation every publisher in $P_b$ receives
$1/(n|P_b|)$---which is the \EQUAL{} rule exactly.

\begin{proposition}[Random order gives the equal split]
\label{prop:random}
With $u$ linear and $\pi$ uniform over orderings of the covering
publishers, $\credit_p=\credit^{\EQUAL{}}_p$ for every $p$.
\end{proposition}

\noindent
\PROV{} is the \emph{same} mechanism with the order dial turned to
history: $\pi$ is chosen to order articles by the time at which they are published. Recall that under linear $u$ the credit for block $b$ goes entirely to whichever
member of $P_b$ the reader reaches first. When $\pi$ is chronological
that member is the earliest publisher of $b$ by construction, so
\eqref{eq:unified} pays $1/n$ to the originator of each block and nothing
to its later publishers, yielding the  $\PROV{}$ rule.

\begin{proposition}[Chronological order gives the provenance rule]
\label{prop:chrono}
With $u$ linear and $\pi$ fixed to be the chronological order of
publication, $\credit_p=\credit^{\PROV{}}_p$ for every $p$.
\end{proposition}

\noindent
Thus the two
rules differ only in their distribution over $\pi$.
One can interpret a uniformly random order $\pi$ as a veil of ignorance: it says the market should not reward
a publisher for reporting on a fact first, only for being able able to supply the fact at all.
Chronological order says the opposite:  being first is what carries value. Intermediate rules, taking distributions over $\pi$ that favor provence without insisting on it for example, provide an interpolation between these extremes.  The choice of distribution is a normative choice about whether original reporting
or availability is what the marketplace exists to fund. In our empirical study, we find that choosing one extreme or the other can shift aggregate revenue share of a publisher by several percentage points. See Section~\ref{sec:results}.

\subsection{Varying the utility model: parameter $\alpha$}
\label{sec:utilitydial}

Suppose that the order $\pi$ is chosen uniformly at random and consider varying the parameter $\alpha$. The exponent
$\alpha$ of \eqref{eq:utility} runs from $\alpha=1$, where every
additional fact is worth the same to the buyer, to the limit in which $u$
tends to $\mathbbm{1}\{x=1\}$ and only a \emph{complete} answer is worth
anything. At that limit a publisher's marginal contribution is one
precisely when it is the arrival that completes the set, and zero
otherwise.  Coupling $\sigma$ from the definition of $\credit^{\COMPL{}}_p$  and $\pi$ and noting that $v_{all}(S)=u(\{b\vert P_b\cap S\not=\emptyset \})$ thus implies the following proposition.

\begin{proposition}[All-or-nothing utility gives the completer rule]
\label{prop:complete}
With $\pi$ uniform over orderings, $\credit_p\to\credit^{\COMPL{}}_p$ for
every $p$ as $\alpha\to\infty$.
\end{proposition}

\noindent
\EQUAL{} and \COMPL{} are therefore the two ends of one continuum rather
than alternatives, and intermediate values of $\alpha$ trace the path
between them. Figure~\ref{fig:alpha}, in Section~\ref{sec:results},
follows that path through our data. The exponent is the buyer's answer to
a single question---\emph{how much is a partial answer worth?}---and the
two rules represent two extreme points.

The space of parameters suggests four corner mechanisms. We have already covered \EQUAL{}, \PROV{} and \COMPL{}. The fourth rule, chronological order with an all-or-nothing buyer, pays the
publisher whose arrival completes the answer, where arrivals are ordered
by publication time. It therefore rewards whoever was slowest among those
the answer could not do without. We report its payouts on our empirical data in Section~\ref{sec:results}.

\subsection{Connection to Shapley value}
\label{sec:shapley}

In a cooperative game, a characteristic function describes the value of any coalition, or subset, of players.  The Shapley value of a cooperative game assigns credit to individuals by considering their marginal contribution to the characteristic function over a uniformly-random subset of players.  It is the unique rule that satisfies {\em efficiency}  (the sum of players' credits equals the total value of the grand coalition), {\em symmetry} (if two players are interchangeable in all coalitions, they receive the same credit), {\em additivity} (for games described by characteristic functions $v$ and $w$, the Shapley values of game that is a linear combination of $v$ and $w$ is the linear combination of the Shapley values of $v$ and $w$), and the {\em null player property} (if a player contributes nothing to any coalition, they get no credit).  As such, it is a popular solution concept in practice, as discussed in Section~\ref{sec:literature}.  

In this section, we show the rules discussed previously are in fact Shapley values of appropriately-defined characteristic functions.\footnote{Note this is non-trivial because Shapley value requires the marginal contribution to be computed over a uniformly-random subset or, equivalently, a uniformly-random permutation.}  Doing so lets us understand the normative consequences of one choice over another by examining the characteristic function.
Recall that \COMPL{} was defined as a
Shapley value in Section~\ref{sec:rules} and Proposition~\ref{prop:complete}
shows the reading model recovers it. The other two rules were derived
from the reading model directly. We now note that they can likewise be interpreted as Shapley values under certain game variations.

Write $v_{\mathrm{lin}}(S)=\frac{1}{n}\,|\{b: S\cap P_b\neq\emptyset\}|$
for the share of the answer's blocks a coalition $S$ can cover. Where
$v_{\mathrm{all}}$ asks whether $S$ can supply everything,
$v_{\mathrm{lin}}$ asks how much of it $S$ can supply.

\begin{proposition}[The equal split is a Shapley value]
\label{prop:shapley}
The \EQUAL{} rule is the Shapley value of $v_{\mathrm{lin}}$.
\end{proposition}

\begin{proof}[Proof sketch]
Write $v_{\mathrm{lin}}=\frac{1}{n}\sum_{b}v_b$ where
$v_b(S)=\mathbbm{1}\{S\cap P_b\neq\emptyset\}$ is a simple game in which only the
members of $P_b$ matter. Any publisher outside $P_b$ is a null player in $v_b$ and receives
nothing; the members of $P_b$ are symmetric in $v_b$; efficiency then
forces each to receive $1/|P_b|$. The result then follows by additivity of the Shapley value. The full proof is in
Appendix~\ref{app:proofs}.
\end{proof}

\PROV{} is likewise a Shapley value, of a third game
which defines the value of a coalition according to what share of the answer the group \emph{broke}. Writing
$p_b$ for the publisher that reported block $b$ first,
\begin{equation}
v_{\mathrm{prov}}(S)\;=\;\frac{1}{n}\,
  \bigl|\{b\in\blocks:\, p_b\in S\}\bigr|,
\label{eq:vprov}
\end{equation}
so a group is credited with a block when one of its members was the first
to report it, and carrying a block somebody else broke is worth nothing.

The consequence is that a publisher which only ever reprints contributes
nothing to \emph{any} coalition, rather than merely to some of them.
Suppose three outlets carry a block and outlet $A$ filed first. Under
$v_{\mathrm{lin}}$ the singleton $\{B\}$ is worth $1/n$: with $A$ absent,
$B$ can still report the block to the reader. Under $v_{\mathrm{prov}}$ the
singleton $\{B\}$ is worth zero, because the block was not $B$'s to break,
and adding $B$ to any group leaves that group's value unchanged. So $B$
is a null player and $A$ is the only publisher whose presence or absence
moves the block at all. The Shapley value pays a null player nothing and
hands the block's whole value to the publisher that alone determines it,
which is the provenance rule.

\begin{proposition}[The provenance rule is a Shapley value]
\label{prop:provshapley}
The \PROV{} rule is the Shapley value of $v_{\mathrm{prov}}$.
\end{proposition}

\noindent
We can interpret $v_{\mathbf{prov}}$ as saying a reprint is worth nothing to the buyer even when the
original is unavailable. Taken literally that is false, and it is exactly
the same commitment the chronological order makes. There the later
publisher earns nothing because it is always reached second, here because
it never counted. This is two ways of saying the same thing: either the
buyer values coverage and reads in the order the news broke, or the buyer
reads in no particular order but values only first publication.
Timestamps enter the mechanism at one point, and whether one puts them
into the reading order or into the game is immaterial; the
credit allocations are the same either way.

\section{Empirical Study}
\label{sec:empirical}

\subsection{The setting and data}
\label{sec:data}

Our data come from a deployed prototype of a block-level marketplace.
Our domain is news articles covering the Iran/US war, available online on June 28, 2026.  A ceasefire had been
in effect at the time, but Iran attacked container ships in the strait of Hormuz during the weekend
and the US retaliated. We selected this episode because it is fast-moving, heavily covered, and
generated a large volume of near-simultaneous reporting by outlets of
very different types, so both original reporting and copying are present
in quantity.

To construct our corpus, we sampled \NQueries{} user queries posed to a major search engine about this news scenario over a 1-hour period on June 28, 2026. Our corpus consists of all articles, records, or URLs returned from grounding in response to these queries.
Of the \NQueries{} user queries, \NQueriesUsed{}
produced an answer citing at least one block; \NGroundings{} grounding
records over \NUrls{} distinct URLs; and \NPublishers{} publishers, of
whom \NPremium{} (\NPremiumPct\%) are classified as premium. Articles
decompose into \NBlockRows{} block instances which cluster into
\NClusters{} distinct blocks. The mean answer cites \MeanCited{} blocks.
Figure~\ref{fig:anatomy} illustrates the resulting distribution of cluster sizes across news and commentary blocks for 
all extracted claims versus those cited in generated answers, for the set of articles crawled.

\begin{figure}[tbp]
\centering
\includegraphics[width=\textwidth]{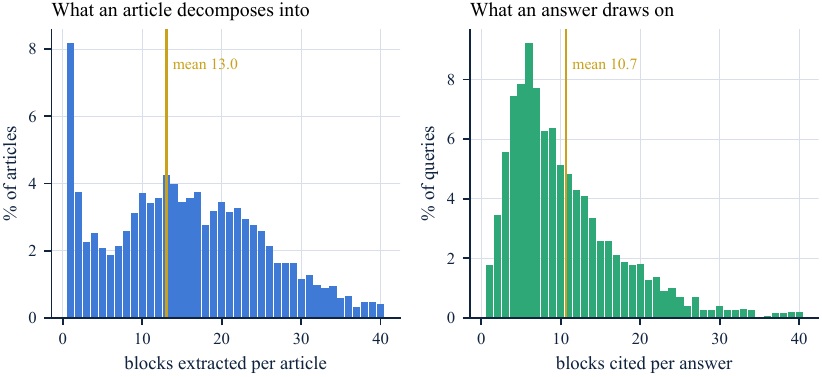}\caption{\textbf{Anatomy of the index.} Left: blocks extracted per article. Right: blocks cited per answer. Both truncated at 40 for legibility.}\label{fig:anatomy}
\end{figure}

\subsubsection{Premium classification}

Throughout we report the share of value accruing to \emph{premium}
publishers.  This is a designation intended to capture established outlets with substantial original
newsgathering operations, in contrast to aggregators, syndicators, and
low-cost republishers.

The premium coding is a hand-curated tier assigned to each publisher domain, with
three levels. \NPremiumUS{} domains are coded \texttt{premium-us} and
\NPremiumIntl{} \texttt{premium-intl}; the remaining \NStandard{} are
\texttt{standard}. The two premium levels are pooled throughout and reported as \emph{premium}, giving
the \NPremium{} publishers (\NPremiumPct\%) that the results tables refer
to. The list is dominated by national newspapers, wire services and
broadcasters.  \NPremiumCountryUS{} of the premium domains are American,
\NPremiumCountryFR{} French, \NPremiumCountryIN{} Indian and
\NPremiumCountryUK{} British.  Importantly, the tier was assigned before any
attribution rule was computed, so it was not fitted to the
results.

The premium classification is not a normative judgment about quality. It is a
proxy for the cost structure the marketplace is meant to affect: those publishers denoted ``premium'' are those that presumably pay reporters to generate content.
As a robustness check, we report
in Appendix~\ref{app:robustness} how the results change when the two
premium levels are separated and when the premium tier is replaced by the
fully mechanical \emph{type} coding, which identifies wire services, national
news, and broadcast outlets.

Two caveats about the set of publishers used in our data. First, publishers are
identified by domain, so a small number of outlets appear more than once
under country editions: \NDupPublishers{} common names are carried by
two or more domains. This splits their credit and therefore mildly
\emph{understates} concentration under every credit attribution rule. Second,
\NTypeOther{} of the \NPublishers{} domains fall in the residual
\emph{other} type under our type coding, which is a reminder that the long tail of the crawl over news articles is primarily covered by small publishers that are not indexed as major news organizations.

\subsubsection{Publication timestamps}
\label{sec:timestamps}

The \PROV{} rule requires knowing the order of publication.
Of the \NUrls{} grounded URLs, \AnyTimeAll\% carry some publication time.
But a large fraction carry only a \emph{date}, which the crawler records
as midnight. \NMidnight{} URLs have a recorded time but it sits at exactly 00:00:00, suggesting a similar truncation, and a further
\NNoTime{} have no time at all. The share carrying an exact,
minute-resolution time is \ExactPrem\% for premium publishers and
\ExactNonPrem\% for the rest. 

We emphasize that this the asymmetry runs \emph{against}
premium outlets in the \PROV{} rule, since they are the group whose exact filing times are least
often recorded. Furthermore, because
\PROV{} is decided by an ordering, and because a third of the relevant
timestamps are imputed, our results for \PROV{} are potentially sensitive to the imputation convention.
We therefore consider and compare four conventions in our reported results: times as stored; date-only records
imputed to noon rather than midnight; all times truncated to the calendar
day so that same-day publications tie; and exact times only, with blocks
lacking any exactly-timed publisher split equally.

\subsection{Empirical Results}
\label{sec:results}

\subsubsection{Prevalence of duplication among blocks}
\label{sec:duplication}

Among cited blocks, evaluated against the global index,
the mean number of carrying publishers is \MeanCovGlobal{} and only
\SingletonGlobal\% have a single publisher. 
Duplication is also unevenly distributed across the two block kinds. News
blocks are singletons \SingletonNews\% of the time, commentary blocks
\SingletonComment\% of the time, meaning that facts get copied more often than the
commentary written about them. Figure~\ref{fig:copygap} reports the
full distributions. 

Due to the prevalence of duplication, we find that deduplicating facts substantially reduced context length. 
Across all \NQueries{} queries in our evaluation, generating answers from block trees 
consumed \TokBoxPrompt{} prompt tokens compared to \TokTextPrompt{} for the raw-text baseline---a \TokPromptCut\% 
reduction in prompt overhead. Including generated completions, total token consumption fell by 
\TokTotalCut\% (Figure~\ref{fig:tokens}).

\begin{figure}[tbp]
\centering
\includegraphics[width=\textwidth]{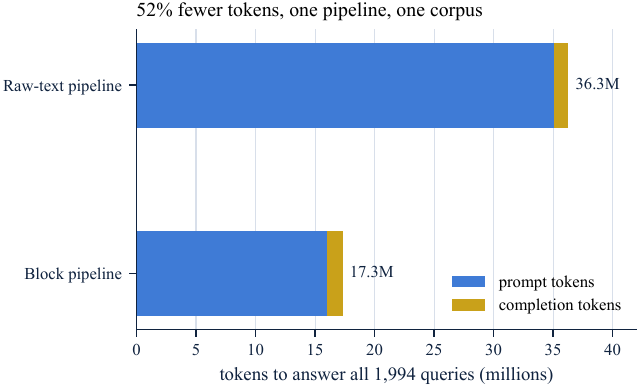}\caption{\textbf{The answering cost ($C_2$).} Prompt and completion tokens per query, comparing deduplicated block trees against raw article text across \NQueries{} queries.}\label{fig:tokens}
\end{figure}

\begin{figure}[tbp]\centering
  \includegraphics[width=\textwidth]{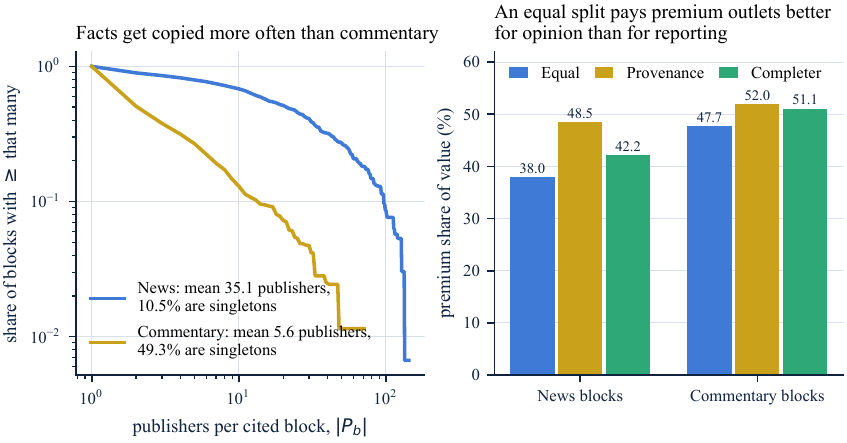}
  \caption{\textbf{The duplication gap.} Left: survival function of publishers per
  cited block in the global index, news against commentary. Right: the same
  gap in payment terms---an equal split pays premium outlets nearly ten
  points more for opinion than for reporting.}\label{fig:copygap}
\end{figure}

\subsubsection{Provenance favors premium publishers}
\label{sec:headline}

Table~\ref{tab:headline} reports the premium share of value under the \EQUAL{} and \PROV{} rules,

computed against the global index,
under each of the four timestamp conventions of
Section~\ref{sec:timestamps}.

\begin{table}[t]
\centering
\begin{threeparttable}
\caption{Premium share of value, global index, by timestamp convention.}
\label{tab:headline}
\begin{tabular}{lcc}
\toprule
Timestamp convention & \EQUAL{} & \PROV{}\\
\midrule
(no timestamps used)          & \EqGlobal\% & ---\\
Times as stored               & --- & \ProvAsis\%\\
Date-only imputed to noon     & --- & \ProvNoon\%\\
Truncated to calendar day     & --- & \ProvDayties\%\\
Exactly-timed publishers only  & --- & \ProvExact\%\\
\bottomrule
\end{tabular}
\begin{tablenotes}\footnotesize
\item \NQueriesUsed{} queries. Ties are split equally under every
convention. \EQUAL{} does not use timestamps and so is invariant to the
convention.
\end{tablenotes}
\end{threeparttable}
\end{table}

Rewarding first publication moves between five and eight percentage
points of total value toward premium publishers. The sign and
magnitude are roughly constant across timestamp conventions. The mechanism
is priority: premium outlets are disproportionately the ones who
report a fact first and are then carried by others, so a rule that prices
availability and a rule that prices priority reach materially
different answers.

\subsubsection{The impact of global indexing}
\label{sec:scoperesults}

Table~\ref{tab:scope} recomputes the same quantities as Table~\ref{tab:headline} but using the local index instead of the global index.

\begin{table}[t]
\centering
\begin{threeparttable}
\caption{The global and local index compared.}
\label{tab:scope}
\begin{tabular}{lcc}
\toprule
 & Local index & Global index\\
\midrule
Mean publishers per cited block & \MeanCovLocal{} & \MeanCovGlobal{}\\
Single-publisher blocks         & \SingletonLocal\% & \SingletonGlobal\%\\[0.3em]
\EQUAL{}, premium share         & \EqLocal\% & \EqGlobal\%\\
\PROV{}, premium share          & \ProvLocalRange\% & \ProvRange\%\\
\bottomrule
\end{tabular}
\begin{tablenotes}\footnotesize
\item \PROV{} ranges are across the four timestamp conventions of
Table~\ref{tab:headline}.
\end{tablenotes}
\end{threeparttable}
\end{table}

Inside the local index, \SingletonLocal\% of cited blocks have exactly one
publisher and mean coverage is \MeanCovLocal{}. 
Since duplication is rare, all rules arrive at approximately the same outcome. As a result,
every rule and every timestamp convention lands
within a single percentage point of every other.
Figure~\ref{fig:scope} shows the
two coverage distributions side by side.

\begin{figure}[tbp]\centering
  \includegraphics[width=\textwidth]{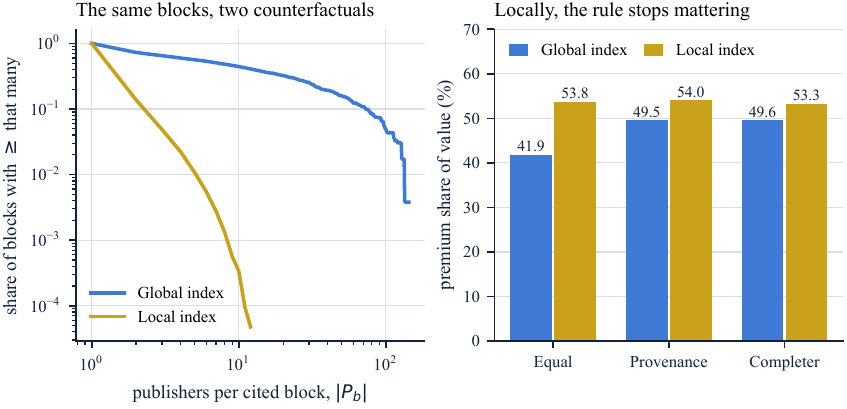}
  \caption{\textbf{Global against local index.} Left: coverage distributions
  under the two scopes. Right: the premium share under each rule. Locally,
  the rules collapse onto each other---there is nothing left for a payment
  rule to disagree about.}\label{fig:scope}
\end{figure}

We regard this as an important empirical result.
At the scope of the local index, \textbf{no attribution policy is expressible at all}.
Since an inference-time method can only ever observe the retrieved
set, different attribution strategies evaluated at inference time against the local index will result in very similar payouts to publishers.

\subsubsection{Comparing the Equal and Completer Rules}
\label{sec:alpharesults}

Next we compare the \EQUAL{} and \COMPL{} rules.  Payments
are the expected marginal contributions as described in Section~\ref{sec:mechanism},
estimated by sampling $200$ publisher orderings per query.
The results are reported in Table~\ref{tab:alpha} and Figure~\ref{fig:alpha}.

\begin{table}[t]
\centering
\begin{threeparttable}
\caption{The utility dial: premium share of value, global index, random order.}
\label{tab:alpha}
\begin{tabular}{lcc}
\toprule
Buyer's utility & Premium share & Publisher-level distance\\
                &                & from \EQUAL{} (TV)\\
\midrule
$\alpha=1$ \ (\EQUAL{}, closed form) & \EqGlobal\%  & ---\\
$\alpha=1$ \ (sampled)               & \AlphaOne\%  & \McNoise{} pp\\
$\alpha=2$                           & \AlphaTwo\%  & \\
$\alpha=4$                           & \AlphaFour\% & \\
$\alpha=8$                           & \AlphaEight\% & \\
$\alpha\to\infty$ \ (\COMPL{})       & \ComplGlobal\% & \TVEqComplGlobal{} pp\\
\bottomrule
\end{tabular}
\begin{tablenotes}\footnotesize
\item Distances are total variation between normalized publisher payout
vectors, computed over all publishers, so a distance can be large even
where the premium aggregate agrees: individual errors cancel within the
group. The second row is the Monte Carlo estimate of a quantity the
first row computes exactly; the gap between them is sampling noise and
bounds the precision of the rest of the column.
Appendix~\ref{app:converge} reports the full convergence, which follows
the textbook $m^{-1/2}$ rate.
\end{tablenotes}
\end{threeparttable}
\end{table}

Premium share rises monotonically in $\alpha$,
from \EqGlobal\% to \ComplGlobal\%. Most of the movement happens
early: going from $\alpha=1$ to $\alpha=2$ buys more than half the total
change. Curvature in the buyer's utility is a preference for
completeness, completeness is delivered by whoever holds the blocks
nobody else has, and premium publishers hold disproportionately many of
those. %

It is worth noting that both \COMPL{} and \PROV{} increase payments to premium publishers relative to \EQUAL{}, and by nearly the same amount in aggregate.
\COMPL{} pays premium publishers \ComplGlobal\% and \PROV{} pays them
\ProvAsis\%, a difference of less than a tenth of a point. These are
economically distinct rules---one rewards being irreplaceable, the other
rewards being early---and they agree on aggregate share to premium publishers. They do not,
however, agree on how payoffs are distributed among individual publishers. The total variation between their
publisher-level payout vectors is \TVProvComplGlobal{} percentage points,
larger than the distance from either of them to \EQUAL{}. 

\begin{figure}[tbp]\centering
  \includegraphics[width=0.85\textwidth]{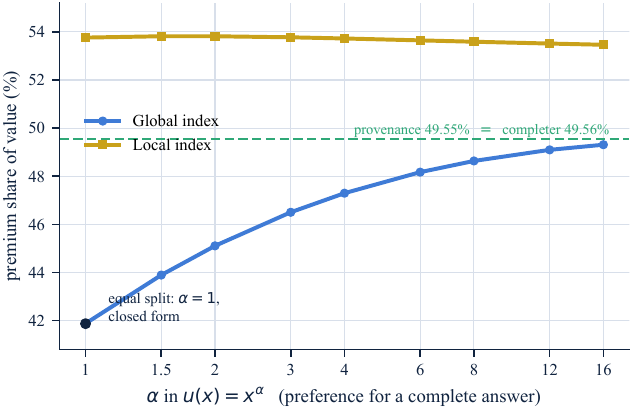}
  \caption{\textbf{The utility dial.} Premium share of value against $\alpha$
  in $u(x)=x^{\alpha}$, random order, under both index scopes. Provenance
  and the completer land within a hundredth of a point of one another and
  are drawn as a single reference line---from opposite economic
  directions.}\label{fig:alpha}
\end{figure}

\subsubsection{The fourth corner, and the two random orders}

Two further outcomes are worth discussing. Combining
the chronological order with the all-or-nothing utility gives the fourth
corner---pay whoever closes the last gap, in historical order---which
pays premium publishers \FourthGlobal\%, inside the range spanned by the
three named rules.

Separately, Section~\ref{sec:orderdial} distinguished two ways of
randomizing: over \emph{articles} and over \emph{publishers}. We define the \EQUAL{} rule as randomizing over publishers. Empirically the distinction is
immaterial in the local index, where the two differ by
\TVArtPubLocal{} percentage points, and not immaterial in the global
index, where they differ by \TVArtPubGlobal{} points and the premium
share moves from \EqGlobal\% to \RandArtGlobal\%. Randomizing over
articles gives extra weight to publishers who filed many stories carrying
the same fact, which is a form of volume subsidy.
This suggests that article-level randomization
could implicitly reward output and repetition rather than content.

\subsubsection{The news and commentary split}

The premium share differs sharply by block kind, and the difference
tracks the duplication gap of Section~\ref{sec:duplication}.

\begin{table}[t]
\centering
\begin{threeparttable}
\caption{Premium share of value by block kind, global index.}
\label{tab:kinds}
\begin{tabular}{lcc}
\toprule
 & News & Commentary\\
\midrule
Cited blocks                 & \NBlocksNews{}   & \NBlocksComment{}\\
Mean publishers per block    & \MeanCovNews{}   & \MeanCovComment{}\\
Single-publisher blocks      & \SingletonNews\% & \SingletonComment\%\\[0.3em]
\EQUAL{}, premium share      & \EqNews\%    & \EqComment\%\\
\PROV{}, premium share       & \ProvNews\%  & \ProvComment\%\\
\COMPL{}, premium share      & \ComplNews\% & \ComplComment\%\\
\bottomrule
\end{tabular}
\end{threeparttable}
\end{table}

A cited news block is carried by \MeanCovNews{} publishers on average; a
cited commentary block by \MeanCovComment{}. 
Under \EQUAL{}, premium publishers
receive \EqNews\% of the value of news blocks but \EqComment\% of the
value of commentary---an equal split pays them substantially better for
opinion than for reporting, because opinion is not copied and so is not
diluted. Moving to \PROV{} raises the news figure to \ProvNews\% while
raising the commentary figure only to \ProvComment\%, which closes most of
the gap.

This has implications for the impact of rule choice on the relative attribution of credit across content types.  
\textbf{An equal split does not merely fail to reward
original reporting; it downweights reporting relative to commentary,
because the property that makes reporting valuable---that others find it
worth copying---is the same property that dilutes its payment.} 

\subsubsection{Concentration}
\label{sec:concentration}

We next analyze how widely payments are dispersed among publishers. This is reported in Table~\ref{tab:conc}. 
The top fifty publishers take \TopFiftyEqual\% of the pool under
\EQUAL{} and \TopFiftyProv\% under \PROV{}. We use the Gini coefficient to measure concentration of payments, and find that it rises from
\GiniEqual{} to \GiniProv{}. Rewarding originality is therefore not a
transfer from the long tail to premium outlets in the aggregate sense
only, but it is also a substantial concentration of the pool. The aggregate premium
share understates this concentration because the movement is toward a small number of
large originators rather than toward the premium tier as a class.

\begin{table}[t]
\centering
\begin{threeparttable}
\caption{Concentration of the pool across publishers, global index.}
\label{tab:conc}
\begin{tabular}{lccc}
\toprule
 & \EQUAL{} & \PROV{} & \COMPL{}\\
\midrule
Top 10 publishers' share  & \TopTenEqual\%   & \TopTenProv\%   & \TopTenCompl\%\\
Top 50 publishers' share  & \TopFiftyEqual\% & \TopFiftyProv\% & \TopFiftyCompl\%\\
Gini coefficient          & \GiniEqual{}     & \GiniProv{}     & \GiniCompl{}\\
Publishers paid nothing   & \ZeroEqual\%     & \ZeroProv\%     & \ZeroCompl\%\\
\bottomrule
\end{tabular}
\begin{tablenotes}\footnotesize
\item Base is the \NPubBase{} publishers who cover at least one cited
block; publishers earning nothing under a rule remain in the denominator,
so the columns are comparable. Under the local index the three rules are
near-indistinguishable, with \EQUAL{} at a Gini of \GiniEqualLocal{} and
\PROV{} at \GiniProvLocal{}.
\end{tablenotes}
\end{threeparttable}
\end{table}

Under \EQUAL{}, every publisher who covers a cited block is paid
something. Under \PROV{}, \ZeroProv\% of them are paid nothing at all,
because being second is worth exactly zero. This is not a distributional
nuance; it is a participation question. A marketplace in which a third of
contributing publishers receive no payment is one that a third of
publishers have no reason to join, and the value of the index to the
buyer depends on their joining.

The completer rule concentrates for a different 
reason, as discussed in Section~\ref{sec:reservation}. Because
$v_{\mathrm{all}}$ is one only when every cited block is covered, a
publisher that uniquely supplies some cited block is a veto player, and
veto players collect \VetoGlobal\% of the pool under the global index.
Under the local index, where \SingletonLocal\% of blocks have a single
publisher, they collect \VetoLocal\%. The completer rule is thus close to
being a rule that pays only for exclusivity, and the local index turns it
into one almost exactly. This is our main reason for treating it as an extremal example within our framework rather than as a candidate policy.

Figure~\ref{fig:lorenz} reports the full distribution of value across
publishers under each rule, with the premium tier overlaid, and
Figure~\ref{fig:concentration} reports the veto-player share alongside
the top-$k$ curves.

\begin{figure}[tbp]\centering
  \includegraphics[width=\textwidth]{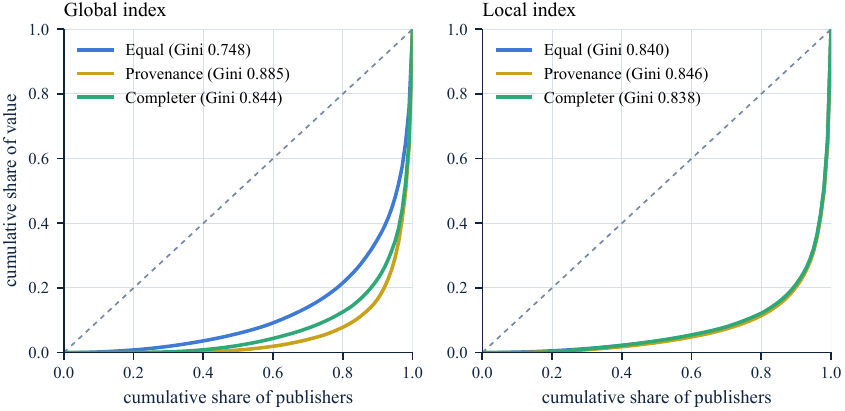}
  \caption{\textbf{Concentration.} Lorenz curves of value across publishers
  under each rule, on a fixed base of contributing publishers. Under the
  local index the three curves are indistinguishable.}\label{fig:lorenz}
\end{figure}

\begin{figure}[tbp]\centering
  \includegraphics[width=\textwidth]{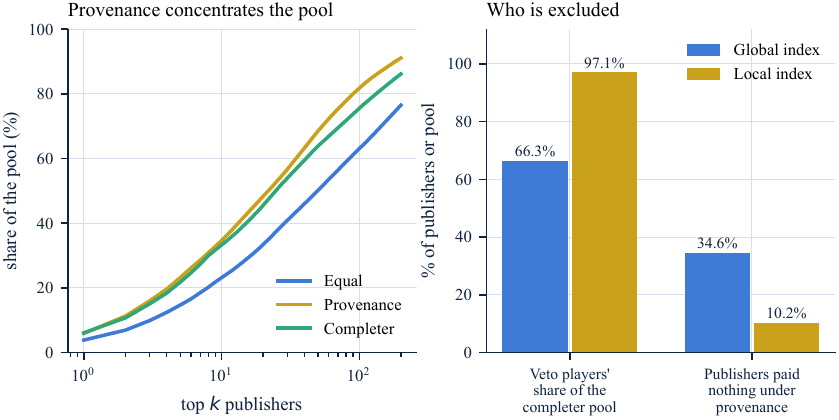}
  \caption{\textbf{Who gets paid.} Left: top-$k$ share of the pool under each
  rule, global index. Right: the share of the completer pool collected by
  veto players, and the share of contributing publishers paid nothing at all
  under provenance.}\label{fig:concentration}
\end{figure}

\subsubsection{Publisher-level results}
\label{sec:namedpubs}

So far, all results have been reported in aggregate. Table~\ref{tab:toppubs} reports
anonymized publisher-level results, since the rules are easier to judge
against individual payout patterns than against a Gini coefficient.

\begin{table}[t]
\centering
\begin{threeparttable}
\caption{The ten largest recipients under \EQUAL{}, and what the same
anonymized publisher groups are paid under the other two rules.}
\label{tab:toppubs}
\begin{tabular}{lccc}
\toprule
Publisher & \EQUAL{} & \PROV{} & \COMPL{}\\
\midrule
Premium publisher/broadcast \#1 & 3.88 \ (1)  & 5.35 \ (2)  & 6.12 \ (1)\\
Premium publisher/broadcast \#2 & 3.07 \ (2)  & 3.50 \ (5)  & 4.62 \ (2)\\
General-interest platform \#1    & 2.95 \ (3)  & 1.49 \ (14) & 4.33 \ (3)\\
Premium publisher/broadcast \#3 & 2.56 \ (4)  & 5.93 \ (1)  & 3.37 \ (4)\\
Standard publisher \#1          & 2.24 \ (5)  & 1.72 \ (12) & 3.31 \ (5)\\
Premium publisher/broadcast \#4 & 1.92 \ (6)  & 3.00 \ (6)  & 2.66 \ (7)\\
Premium publisher/broadcast \#5 & 1.88 \ (7)  & 2.43 \ (7)  & 2.63 \ (8)\\
Premium publisher/broadcast \#6 & 1.70 \ (8)  & 1.92 \ (9)  & 1.74 \ (9)\\
Premium publisher/broadcast \#7 & 1.63 \ (9)  & 1.18 \ (17) & 1.17 \ (16)\\
Premium publisher/broadcast \#8 & 1.32 \ (10) & 2.23 \ (8)  & 1.29 \ (14)\\
\addlinespace
\multicolumn{4}{l}{\emph{Outside this ten, but inside the top ten under one of the other rules}}\\
\addlinespace
Premium publisher/broadcast \#9  & \NYTEq{} \ (\NYTEqRank{}) & \NYTPv{} \ (\NYTPvRank{}) & 0.66 \ (30)\\
Specialized information site \#1 & \MilEq{} \ (\MilEqRank{}) & \MilPv{} \ (\MilPvRank{}) & 0.16 \ (99)\\
Reference platform \#1            & 0.92 \ (18) & 1.88 \ (10) & 1.03 \ (22)\\
Premium publisher/broadcast \#10 & \ReutersEq{} \ (11) & 1.31 \ (16) & \ReutersCm{} \ (6)\\
Premium publisher/broadcast \#11 & 1.10 \ (14) & 1.37 \ (15) & 1.64 \ (10)\\
\bottomrule
\end{tabular}
\begin{tablenotes}\footnotesize
\item Percent of the whole pool, global index, with \PROV{} under times as
stored. Rank among the \NPubBase{} contributing publishers in parentheses.
Labels are anonymized for this version of the paper.
\end{tablenotes}
\end{threeparttable}
\end{table}

The top of the equal split is held by outlets that publish prolifically.
Premium publisher/broadcast \#1, premium publisher/broadcast \#2,
general-interest platform \#1, premium publisher/broadcast \#3 and
standard publisher \#1 lead it. That is the rule working as designed: it
pays for having the fact, and these are the outlets that have the most facts.

Premium publisher/broadcast \#9 is the cleanest illustration of the distinction between \EQUAL{} and \PROV{}. 
It is \NYTEqRank{}th under \EQUAL{} with \NYTEq\% of the
pool and \NYTPvRank{}rd under \PROV{} with \NYTPv\%, a factor of
\NYTRatio{}. The mechanism is the dilution of
Section~\ref{sec:duplication}: facts that an outlet originates may be
carried by many other outlets, so an equal split divides each of them
twenty-odd ways. General-interest platform \#1 moves the other way, from
\YahooEq\% to \YahooPv\% and from third to fourteenth. Premium
publisher/broadcasts \#10 and \#11 rise under \COMPL{} without rising
under \PROV{}---\#10 from \ReutersEq\% to \ReutersCm\%---which is the
signature of holding material nobody else carried rather than of having
been early.

Another row from the table is worth highlighting. Specialized information site \#1 is
\MilEqRank{}th under \EQUAL{} with \MilEq\% of the pool and
\MilPvRank{}th under \PROV{} with \MilPv\%, which puts a single-purpose
site ahead of almost every newsroom in the corpus. We do not believe this specialized information site was the first to report these facts.
We believe it carries filing times that sort early,
as discussed in Section~\ref{sec:timestamps}.
This illustrates the potential fragility of a first-mover rule:
it is only as good as the timestamps it is decided by.

\subsection{Empirical Study Limitations}
\label{sec:limitations}

Our evidence comes from a single news episode served by a single
retrieval and generation stack. The Iran/US coverage was chosen
because heavy simultaneous reporting makes duplication measurable, but
that same feature may make our covering sets unusually large. We would
expect the qualitative results---that duplication is common and that the
index scope dominates the rule---to generalize.

Publication timestamps are substantially imputed. As
Section~\ref{sec:timestamps} sets out, only \ExactPrem\% of premium and
\ExactNonPrem\% of non-premium grounded URLs carry an exact time. As a robustness check, we report four different conventions and verify that our results are largely stable across them.
That said, we note that this robustness may not generalize to other datasets, so further study on the impact of imputation methods is warranted.  

A deeper point about ordering is that we observe \emph{who filed first},
which is not the same as \emph{who found out first}. Two outlets may work
the same story independently and file an hour apart; one may lift the
other's reporting wholesale and file a minute later. The index records
both as one cluster with two publishers and a filing order, and nothing in
our data distinguishes independent discovery from copying. \PROV{}
therefore prices \emph{priority}, not \emph{originality}, and reads as a
rule about original reporting only under the assumption that the two
coincide. That assumption is plausible in a fast-moving news
episode, where being first usually does mean having done the work, but it
is an assumption.

The index is also bounded by the crawl, and our estimates condition on
retrieval. The global index is global relative to what was crawled, not
relative to everything published. The direction of this bias is signable:
a larger corpus can only increase the mean number of publishers per block,
which can only widen the gap between rules that price availability and
rules that price originality. Our estimates of that gap are therefore
best read as lower bounds. Separately, all our rules are computed over
blocks that were actually \emph{cited}, so they still condition on the
retriever's choices. The framework does not require this---covering sets
are defined over the whole index---but our measurements inherit it.

A further constraint is extraction and clustering error. The block index
is produced by a language model and inherits its mistakes, both in
decomposing articles and in deciding that two blocks express the same
content. Under-clustering inflates the singleton share and biases toward
rules that reward indispensability; over-clustering does the reverse.
Appendix~\ref{app:robustness} perturbs the index directly along all three
error modes---misfiled kinds, clusters wrongly split, clusters wrongly
merged---at rates up to \ErrMaxRate\%. We view this as a reasonable upper bound on plausible
error rate for this pipeline. Our main results are orobust to these perturbations: the
\PROV{}--\EQUAL{} spread stays between \ErrSpreadLo{} and \ErrSpreadHi{}
percentage points against a baseline of \SpreadPooled{}. 

Finally, certain potentially relevant parameters are not observable in our dataset. We take answer quality
as given and make no claim that block-based generation produces better
answers, only cheaper ones. We do not model publishers' entry or
investment decisions, so %
any discussion of publisher incentives is informal rather than an equilibrium analysis.
We also do not observe
prices: we allocate shares of a normalized pool, leaving the level of the
pool---which is what publishers actually care about---entirely outside
the paper.

\section{Conclusion}
\label{sec:conclusion}

The click is disappearing as a way to compensate publishers for their content, 
and something has to replace it. While there is broad agreement that publishers should be
paid for the content ingested and served by AI to answer queries, the norms and mechanisms by which payments are dispersed have not yet been settled.

We have argued three things. First, a variety of natural credit attribution rules---equal split, first-mover, and completer--emerge as Shapley values in a general framework of news readership games.

Second, we have argued that the distinction between these rules is meaningful in empirical data. 
We find that attributing credit based on a first-mover rule rather than equal split shifts five to eight percentage points
of value toward premium publishers, robustly across four ways of handling
incomplete timestamps. This suggests a non-trivial impact on usage payments, though arguably not a complete transformation of investment incentives.

Third, we explored the impact of attributing credit with respect to 
a \emph{global index} of everything published, at index time, versus evaluating credit against
the \emph{local index} of sources retrieved for one query at inference time.
In the global index a cited news fact has a mean
of \MeanCovGlobal{} publishers and the rules separate. In the local index
\SingletonLocal\% of cited facts have a single publisher and every rule
results in essentially identical payments. We conclude that a global index is critical for designing allocation policies that meaningfully capture an intended metric of publisher contribution.
Attribution
computed when a query arrives can only ever see the local index, so it can
never price the counterfactual the market cares about---what this
publisher is worth given everything else that exists. Attribution computed
when a document is indexed can, with the architectural advantage that this determination does not impose a significant computational cost per query. 

Finally, a note on what a marketplace settlement layer requires that a metering
scheme does not. A global indexing system can be observed and verified: Publishers can be
shown their blocks, who else
carried them, and all timestamps. This observability enables contracts written on these verifiable
states. The auditability of
an index-based measurement system therefore enables contractable payments to publishers along with the feedback needed to help publishers invest in content that is valued.  We argue that this is an important step toward developing marketplaces and mechanisms that maintain the long-term health and stability of the content ecosystem in the age of AI-generated answers.

\clearpage
\bibliography{refs}

\clearpage
\appendix
\section{Proofs}
\label{app:proofs}

\subsection*{Proposition~\ref{prop:random}}

Fix a block $b$ with covering set $P_b$ and let $\pi$ be uniform over
orderings of the publishers. Under linear $u$, the increment in
\eqref{eq:unified} contributed by block $b$ is $1/n$ for whichever member
of $P_b$ appears earliest in $\pi$ and zero for all others. By exchangeability
each member of $P_b$ is earliest with probability $1/|P_b|$, so the expected
credit from $b$ is $1/(n|P_b|)$ for each $p\in P_b$ and zero otherwise.
Summing over blocks and using linearity of expectation gives
$\credit_p=\frac{1}{n}\sum_b \mathbbm{1}\{p\in P_b\}/|P_b|=\credit^{\EQUAL{}}_p$.
\qed

\subsection*{Proposition~\ref{prop:chrono}}

Fix a block $b$ and let $\pi$ place the publishers in the order in which
they published. Under linear $u$ the increment in \eqref{eq:unified}
contributed by $b$ is $1/n$ for whichever member of $P_b$ comes first in
$\pi$ and zero for the rest, exactly as in
Proposition~\ref{prop:random}; the only change is that $\pi$ is now a
point mass rather than uniform, so no expectation is taken. The earliest
member of $P_b$ under a chronological $\pi$ is $f_b$, the first publisher
of $b$. Summing over blocks,
$\credit_p=\frac{1}{n}\sum_b\mathbbm{1}\{f_b=p\}=\credit^{\PROV{}}_p$. \qed

\subsection*{Proposition~\ref{prop:complete}}

The coverage fraction $c_p(\pi)/n$ takes values in the finite set
$\{0,1/n,\dots,1\}$, on which $x^{\alpha}\to\mathbbm{1}\{x=1\}$ pointwise as
$\alpha\to\infty$. The increments in \eqref{eq:unified} lie in $[0,1]$
and $\pi$ ranges over a finite set, so the limit passes through the
expectation. In the limit the increment for $p$ under $\pi$ is
$\mathbbm{1}\{c_p(\pi)+a_p(\pi)=n\}-\mathbbm{1}\{c_p(\pi)=n\}$, which is one exactly
when the publishers ahead of $p$ in $\pi$ leave some block uncovered
and those up to and including $p$ cover every block. That is the marginal
contribution of $p$ to $v_{\mathrm{all}}$ at the predecessor set induced
by $\pi$, and averaging it over uniform $\pi$ is the definition of the
Shapley value of $v_{\mathrm{all}}$, which is $\credit^{\COMPL{}}$. \qed

\subsection*{Proposition~\ref{prop:shapley}}

Let $v_b(S)=\mathbbm{1}\{S\cap P_b\neq\emptyset\}$, so that
$v_{\mathrm{lin}}=\frac{1}{n}\sum_{b}v_b$. Fix $b$. Any $p\notin P_b$
satisfies $v_b(S\cup\{p\})=v_b(S)$ for all $S$ and is therefore a null
player, receiving zero. Any two $p,q\in P_b$ are interchangeable in
$v_b$, so by symmetry they receive equal amounts. Efficiency requires the
payments to sum to $v_b(\pubs)=1$. Hence each $p\in P_b$ receives
$1/|P_b|$. The Shapley value is additive in the game, so the value of
$v_{\mathrm{lin}}$ is $\frac{1}{n}\sum_b$ of these, which is exactly
$\credit^{\EQUAL{}}$. \qed

\subsection*{Proposition~\ref{prop:provshapley}}

Write $v_{\mathrm{prov}}=\frac{1}{n}\sum_b w_b$ with
$w_b(S)=\mathbbm{1}\{f_b\in S\}$. Each $w_b$ is a dictator game: $f_b$
determines the value on its own and every other publisher is a null player, so
efficiency and the null-player axiom force the Shapley value of $w_b$ to
pay $1$ to $f_b$ and $0$ to everyone else. Additivity then gives
$\credit_p=\frac{1}{n}\sum_b\mathbbm{1}\{f_b=p\}=\credit^{\PROV{}}_p$. \qed

\begin{remark}
The symmetry axiom is what separates the two representations. In
$v_{\mathrm{lin}}$ the members of $P_b$ are interchangeable, so any
symmetric value must split $b$ equally among them and \PROV{} is
unreachable; \PROV{} is recovered either by dropping symmetry, which is
what a non-uniform $\pi$ does, or by moving to a game in which the
publishers of $b$ are no longer interchangeable, which is what
$v_{\mathrm{prov}}$ does. Timestamps enter the mechanism at exactly one
point, and this is it.
\end{remark}

\begin{remark}
Note that $v_{\mathrm{all}}=\min_b v_b$. Unlike $v_{\mathrm{lin}}$ it is
not additively separable across blocks, which is why \COMPL{} admits no
closed form and must be evaluated by sampling permutations.
\end{remark}

\section{Robustness}
\label{app:robustness}

\subsection{The premium classification}

The premium tier is the least mechanical construct in the paper, so we
report the headline under four alternative definitions. The concern is
that a hand-assigned prestige coding might be doing the work that we
attribute to priority.

\begin{table}[h]
\centering
\begin{threeparttable}
\caption{Premium share under alternative definitions of the tier, global index.}
\label{tab:tierrobust}
\begin{tabular}{lcccc}
\toprule
Definition & Publishers & \EQUAL{} & \PROV{} & Spread\\
\midrule
Pooled premium (headline)   & \NPremium{}     & \EqGlobal\%    & \ProvAsis\%     & \SpreadPooled{}\\
Type-based newsroom coding  & \NTypeNewsroom{} & \EqTypeBased\% & \ProvTypeBased\% & \SpreadTypeBased{}\\
\texttt{premium-us} only    & \NPremiumUS{}   & \EqUSOnly\%    & \ProvUSOnly\%    & \SpreadUSOnly{}\\
\texttt{premium-intl} only  & \NPremiumIntl{} & \EqIntlOnly\%  & \ProvIntlOnly\%  & \SpreadIntlOnly{}\\
\bottomrule
\end{tabular}
\begin{tablenotes}\footnotesize
\item The type-based coding marks a publisher as a newsroom if its
independently assigned \emph{type} is one of national news, broadcast
news, wire service, business/financial, magazine/long-form or digital
news. It makes no reference to prestige and was not constructed for this
exercise. Spread is $\PROV{}-\EQUAL{}$ in percentage points.
\end{tablenotes}
\end{threeparttable}
\end{table}

The first two rows are the reassuring part. Replacing the prestige tier
with the independent type coding moves the levels by under two points and
the spread by less than a tenth of a point, from \SpreadPooled{} to
\SpreadTypeBased{}. The result is a property of which publishers report
first, not of how we labeled them.

The last two rows are not reassuring, and we think they are the more
interesting finding. The entire provenance premium accrues to American
premium outlets, whose share rises from \EqUSOnly\% to \ProvUSOnly\%, a
spread of \SpreadUSOnly{} points. International premium outlets gain
nothing: their share \emph{falls} slightly, from \EqIntlOnly\% to
\ProvIntlOnly\%, a spread of \SpreadIntlOnly{} points. \textbf{So
\PROV{} does not reward premium publishers as a class. It rewards
American premium publishers, and it is close to neutral between
international premium outlets and the standard tier.}

We do not want to over-read this from a single episode, and there is an
obvious mechanical candidate for at least part of it: the corpus was
crawled from a US-facing product, and the timestamp quality reported in
Section~\ref{sec:timestamps} may differ systematically by region, so
non-US outlets may be losing ties they in fact won. But it is exactly the
kind of incidence a marketplace designer needs to see before adopting a
first-mover rule, and it is invisible in the pooled number.

\subsection{Extraction and clustering error}
\label{app:errors}

The index is produced by a language model, so it is wrong some of the
time. We do not have ground truth against which to measure how often, so
instead we ask how much error the headline could absorb. Three error
modes move the covering set $P_b$ of a cited block in different
directions, and we simulate each directly on the index.

\begin{description}
\item[Misfiled kind] A block is recorded under the wrong kind. The
  publisher leaves $b$ and lands on the block with the same cluster
  identifier and the other kind, while publishers of that sibling leak into
  $b$. The size of $P_b$ is roughly preserved and its composition is
  churned.
\item[Wrongly split] The clusterer fails to see that two publishers said
  the same thing, so a publisher is detached into a cluster of its own.
  That cluster is not the one the answer cited, so the publisher is not
  paid; $|P_b|$ shrinks.
\item[Wrongly merged] The clusterer conflates two distinct claims, so
  publishers who do not carry $b$ are recorded as carrying it; $|P_b|$
  grows.
\end{description}

Each mode is applied independently to every publisher of every cited block
at rate $\varepsilon$, and \EQUAL{} and \PROV{} are recomputed on the
perturbed index. Both rules are closed form per block, so no sampling is
involved; we average over five replicates. At $\varepsilon=0$ the
procedure reproduces the numbers of Section~\ref{sec:headline} exactly,
which is the check that it is measuring what we intend.

\begin{table}[t]
\centering
\begin{threeparttable}
\caption{Extraction- and clustering-error robustness, global index.}
\label{tab:errors}
\begin{tabular}{llccc}
\toprule
Perturbation & Rate & \EQUAL{} & \PROV{} & Spread\\
\midrule
none                & ---   & \EqGlobal\% & \ProvAsis\% & \SpreadPooled{}\\
\addlinespace
Misfiled kind       & 0.05  & \ErrKindFiveEq\%  & \ErrKindFivePv\%  & \ErrKindFiveSp{}\\
                    & 0.20  & \ErrKindTwentyEq\% & \ErrKindTwentyPv\% & \ErrKindTwentySp{}\\
\addlinespace
Wrongly split       & 0.05  & \ErrSplitFiveEq\%  & \ErrSplitFivePv\%  & \ErrSplitFiveSp{}\\
                    & 0.20  & \ErrSplitTwentyEq\% & \ErrSplitTwentyPv\% & \ErrSplitTwentySp{}\\
\addlinespace
Wrongly merged      & 0.05  & \ErrMergeFiveEq\%  & \ErrMergeFivePv\%  & \ErrMergeFiveSp{}\\
                    & 0.20  & \ErrMergeTwentyEq\% & \ErrMergeTwentyPv\% & \ErrMergeTwentySp{}\\
\bottomrule
\end{tabular}
\begin{tablenotes}\footnotesize
\item Mean of five replicates. Spread is $\PROV{}-\EQUAL{}$ in percentage
points. A rate of $0.20$ means one publisher record in five is wrong, which
is far beyond any error rate we would defend for this pipeline; it is
included to show where the result would have to break.
\end{tablenotes}
\end{threeparttable}
\end{table}

Table~\ref{tab:errors} reports the result, and it is dull in the way one
wants a robustness table to be dull. Across all three modes and rates up
to \ErrMaxRate\%, the spread stays between \ErrSpreadLo{} and
\ErrSpreadHi{} points against a baseline of \SpreadPooled{}. The reason
is structural rather than lucky. Splitting and misfiling both thin the
covering sets---mean $|P_b|$ falls from \MeanCovGlobal{} to about
\ErrMeanCovTwenty{} at the highest rate---but they thin the sets of
premium and non-premium publishers in the same proportion, and both rules
are homogeneous of degree zero in a uniform thinning. Only merging moves
the spread at all, and it moves it \emph{up}: conflating claims imports
extra publishers into a block, which dilutes the equal split faster than it
dilutes a first-mover rule whose winner is usually already the earliest
publisher in the enlarged set.

The one caveat worth stating is that our perturbations are uniform, and
real extraction error will not be. If the model is systematically worse
at parsing one class of outlet---shorter wire copy, non-English sources,
paywalled excerpts---the error is correlated with the very tier
distinction the headline is about, and no amount of uniform perturbation
would reveal it. Measuring that requires labeled ground truth, which we
do not have.

\subsection{Order and utility}

Section~\ref{sec:alpharesults} reports the article-level against
publisher-level randomization and the Monte Carlo validation. To
summarize: the two randomizations differ by \TVArtPubGlobal{} percentage
points of total variation in the global index and \TVArtPubLocal{} in the
local index, and the sampled $\alpha=1$ arm reproduces the closed-form
\EQUAL{} payout vector to \McNoise{} points, which is the sampling error
implied
by $200$ permutations per query and bounds the precision of every
sampled quantity we report.

\subsection{Monte Carlo convergence}
\label{app:converge}

Every $\alpha$-indexed quantity in this paper is estimated from $200$
publisher orderings per query. This section asks what that number buys,
and it can ask the question unusually sharply because at $\alpha=1$ the
estimand is known in closed form: by Proposition~\ref{prop:shapley} the
sampler is estimating the equal per-block split, which we can also
compute exactly. So we can measure error against ground truth rather than
only against the sampler itself.

We run \MCReps{} independent streams of \MCMaxPerm{} orderings per query
and read the estimate off prefixes of each stream, which yields every
sample size from one pass. Two quantities are reported, and they answer
different questions. \emph{Accuracy} is the total variation between the
sampled credit vector and the exact one, available only at $\alpha=1$.
\emph{Precision} is the mean pairwise total variation between the
independent streams, which is defined for every rule.

Figure~\ref{fig:converge} shows both. The accuracy curve falls at a
fitted rate of $m^{\MCSlopeTruth}$, and the precision curves at rates
between $m^{\MCSlopeRepLo}$ and $m^{\MCSlopeRepHi}$: the textbook
$m^{-1/2}$, with no sign of a slower regime. Concretely, the distance
from the truth falls from \MCTruthTwentyFive{} points at $m=25$ to
\MCTruthTwoHundred{} at $m=200$ and \MCTruthEightHundred{} at
$m=\MCMaxPerm{}$.

The two measures are related in a way that lets us export the closed-form
benchmark to the rules that have none. If two independent estimates carry
independent errors of comparable size, the distance between them should
exceed the distance from either to the truth by a factor of $\sqrt 2$. At
$\alpha=1$ the observed ratio is \MCRepRatio{}, against $\sqrt 2 \approx
1.414$. Taking that seriously, the between-stream distances at $m=200$
imply errors of roughly \MCErrAlphaTwo{} points at $\alpha=2$,
\MCErrAlphaEight{} at $\alpha=8$ and \MCErrCompl{} for the completer.
The completer is the hardest of the family to estimate, which is what one
expects of a rule that depends on the identity of a single pivotal
publisher rather than on a smooth average.

These magnitudes are what license the readings in
Section~\ref{sec:alpharesults}. The gap of \McNoise{} points
between the closed-form \EQUAL{} rule and its sampled counterpart is
squarely inside the sampling error and should be read as noise, which is
the point of reporting it. By contrast the distances that carry the
paper's argument---\TVEqComplGlobal{} points from \EQUAL{} to the
completer, \TVProvComplGlobal{} points from \PROV{} to the
completer---are more than an order of magnitude larger than the noise
floor, so no plausible increase in $m$ would disturb them.

One asymmetry is worth flagging because it recurs elsewhere in the paper.
The \emph{aggregate} converges far faster than the \emph{allocation}. At
$m=200$ the premium share is within \MCPremErr{} points of its exact
value while the credit vector is still \MCTruthTwoHundred{} points away
in total variation, a factor of about fifty; the premium share is already
accurate to a twentieth of a point at $m=25$. Averaging over a thousand
publishers cancels most of the sampling error, but the individual
payments it cancels are exactly what a marketplace would actually pay
out. A settlement system should therefore not choose $m$ by watching the
headline. The same caution appears substantively in
Section~\ref{sec:alpharesults}, where \PROV{} and the completer agree on
the aggregate to within a hundredth of a point while differing by
\TVProvComplGlobal{} points on who gets paid.

\begin{figure}[tbp]\centering
  \includegraphics[width=\textwidth]{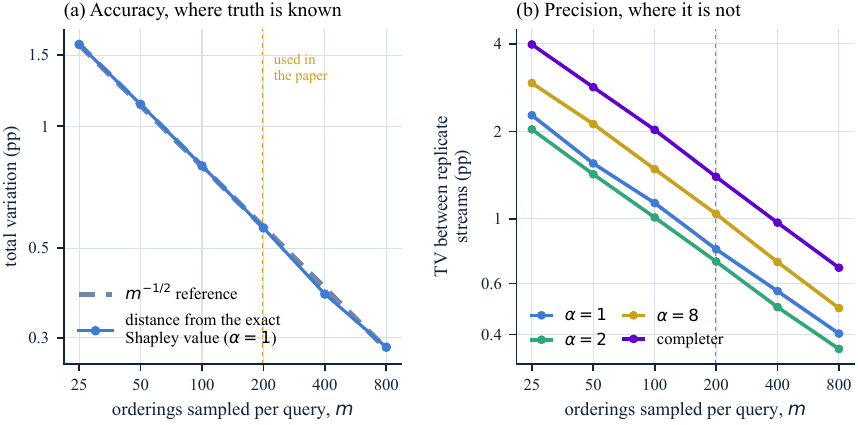}
  \caption{\textbf{Monte Carlo convergence.} Panel (a): distance from the
  exact Shapley value at $\alpha=1$, against a $m^{-1/2}$ reference the
  data are indistinguishable from. Panel (b): mean pairwise distance
  between \MCReps{} independent streams, for rules with no closed form.
  The dashed vertical line marks the $m=200$ used throughout the
  paper.}\label{fig:converge}
\end{figure}

\todo{Still to run: the four timestamp conventions at publisher level
rather than in aggregate.}

\section{Discussion: Global indexing}
\label{sec:globalindex}

In Section~\ref{sec:lego} we described our system for answering queries and attributing credit using the global index. 
We now return to our choice to use a global, pre-computed index and compare it with alternative solutions.

\subsection{Index scope: local and global}

Recall from Section~\ref{sec:answering} that, for any given query, a search engine is used by a retriever system to return a small grounding set of articles. Every rule in Section~\ref{sec:mechanism} is defined with respect to a choice of covering
sets $P_b$ for each block $b$ extracted from the grounding set. The choice of $P_b$ requires deciding which corpus to consult,
and there are two coherent answers.

\begin{description}
\item[The global index.] $P_b$ is every publisher in the index who
  carried block $b$, whether or not their article was retrieved for this
  query. The counterfactual is one of \emph{publication}: what would the
  answer have looked like if this fact had never been reported at all?
\item[The local index.] 
$P_b$ is every publisher who carried block $b$ \emph{in the query's grounding set}---the handful of articles the retriever
  actually surfaced. The counterfactual is one of \emph{retrieval}: what
  would the answer have looked like if we had not shown this particular
  article during grounding?\footnote{The terms collide with an established usage in
  distributed information retrieval, where ``local'' and ``global'' index
  refer to document- and term-partitioned shards of a search index. We
  use them in the economic sense throughout: the scope of the
  counterfactual, not the physical layout of the index.}
\end{description}

\subsection{Attribution timing: Index time and inference time}

A second, related aspect of our proposed appproach concerns \emph{when} attribution is
computed.

\emph{Inference-time attribution} waits for a query and, while generating an answer, works out
which of the retrieved documents caused which parts of the generated
responses. Three existing approaches to credit attribution are performed at inference time: asking the model to cite its sources
\citep{gao2023, bohnet2022}; perturbing the context and re-running, as in
ContextCite \citep{cohenwang2024} or Shapley over retrieved documents
\citep{nematov2025}; and reading the model's internals
\citep{qi2024}.

\emph{Index-time attribution}, which our system employs, calculates attribution as part of the indexing work that occurs before any query exists. Each
article is decomposed into blocks and merged into the index once, when it
is crawled, so every block already carries the list of outlets that
published it and the time each of them filed. When a query arrives,
settlement simply reads the covering sets off the index and divides the fee.
This is the pipeline described in Section~\ref{sec:lego}.

Timing and scope are different distinctions, but they are not
independent. Working out attribution at inference time effectively forces the use of the local index, because
it is impractical to recalculate the global index for every query.
When building a global index in advance, one could still use the local index for attribution. We
argue below, however, that using the global index has substantial advantages for accurate credit attribution.

\subsection{The computational cost of index-time and inference-time attribution}

A practical disadvantage of index-time computation the computational cost of constructing the index in advance. Indeed, indexing is often considered to be expensive compared to inference-time calculation. We note, however, that index construction is an up-front cost that makes both answer generation and credit attribution---both per-query costs---less expensive. Moreover, in our system, the covering set of articles is already materialized in the index, so settlement of credit and payments reduces to set
arithmetic. That is also what makes the \emph{global} index relatively affordable. We therefore expect a global index with index-time attribution to be superior when the number of queries is sufficiently large, as should be the case for a content marketplace that is large enough to meaningfully influence publisher outcomes.

\subsection{Further arguments in favor of index-time computation}

We want to be fair to the inference-time literature: those methods are
excellent at what they are designed for, which is explaining and
verifying a generation. They are the wrong instrument for paying publishers,
for five reasons beyond cost.

The first and most important is that they answer the wrong
counterfactual. Inference-time methods ask which pieces of \emph{this}
context moved the output. The marketplace question is what a publisher is
worth given everything else that exists. These differ whenever a fact
appears in more than one source, which in our data is the common case. A
perturbation method that removes one outlet's document watches the answer
degrade and pays that outlet; it cannot see that three other outlets
reported the same event and any of them would have served.
Section~\ref{sec:results} shows that inside the local index every rule in
the family returns the same answer to within one percentage point, so
this is not a matter of precision: \textbf{inference-time attribution
cannot tell a fact only one outlet had from a fact twenty outlets had,
even in principle, because the other nineteen were discarded before the
query arrived.}

Second, they launder retrieval bias into payment. If the retriever is
popularity- or SEO-driven, whoever is retrieved is paid, and being paid
funds being retrieved. Inference-time attribution has no defense, because
retrieval is upstream of it. A global covering set at least allows a rule
to be defined over all publishers who reported the fact rather than the
$k$ the retriever surfaced. We note honestly in
Section~\ref{sec:limitations} that our own estimates are still computed
over \emph{cited} blocks and so still condition on retrieval; the
framework does not have to.

Third, they are not auditable, and markets require auditability. A
publisher cannot verify a saliency map or a perturbation run: it cannot
see the prompt, cannot re-run the model, and cannot reproduce a sampled
generation. A block index is a persistent, inspectable artifact---these
are your blocks, this is who else carried them, this is when you
filed---against which disputes can be resolved. Contracts must be written
on verifiable states, and an unauditable metering standard is not a
metering standard.

Fourth, they are unstable and gameable. Generation is stochastic and
self-citation is unreliable: \citet{liu2023} find only 51.5\% of
sentences fully supported by their cited sources. Paying on a signal with
that error rate is untenable. Duplication attacks---republishing content
so that more chunks are retrieved---inflate credit under any
per-retrieved-chunk scheme, whereas cross-publisher deduplication
neutralizes them by construction, since a block is credited once however
many documents carry it.

Fifth, provenance is impossible at inference time. First-mover credit
requires knowing who reported a fact first. That is a property of the
corpus's history, not of a single prompt, and no amount of context
attribution recovers it. Any marketplace that wishes to reward original
reporting must therefore maintain something like a block index.

\subsection{Counter-arguments to the global index}

Three concessions are due. First, a global index is susceptible to \emph{extraction error}: index-time
attribution inherits whatever mistakes the decomposition and clustering
make, and inference-time methods have no such dependency. We investigate these errors in our empirical exercise via robustness checks in 
Section~\ref{sec:limitations}. Second, there is a non-trivial cost of \emph{maintaining the index}. This is a recurring bill:
fresh content must be indexed before it can be attributed, which is
a latency the inference-time approach does not have. 
As we argued above, this cost is per-document and may be quite low when amortized over all queries.
But in a setting that requires serving few queries over a large and fast-churning corpus, our global indexing method may not be appropriate.
Third,
\emph{the two methods are complementary}: even when index-time attribution is effective for credit attribution, inference-time attribution remains the
right tool for verifying that a model actually used what it cites. A
mature marketplace would plausibly use inference-time methods to police
the citation step and index-time methods to settle attribution.

\ifDRAFT
  \clearpage
  \listoftodos[\textbf{Outstanding items}]
\fi

\end{document}